\documentclass{article}
\usepackage{arxiv}

\usepackage[utf8]{inputenc} 
\usepackage[T1]{fontenc}    
\usepackage{hyperref}       
\usepackage{url}            
\usepackage{booktabs}       
\usepackage{amsfonts}       
\usepackage{nicefrac}       
\usepackage{microtype}      
\usepackage{lipsum}		
\usepackage{graphicx}
\usepackage{doi}
\usepackage{amsmath}
\usepackage{subcaption}
\usepackage{amsthm}
\newtheorem{theorem}{Theorem}

\title{Stability of Cantilever-like Structures with Applications to Soft Robot Arms}

\author{Siva Prasad Chakri Dhanakoti}

\author{ \href{https://orcid.org/0000-0001-8346-4289}{\includegraphics[scale=0.06]{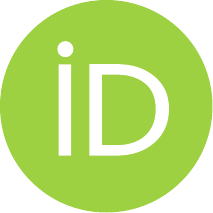}\hspace{1mm} Siva Prasad Chakri Dhanakoti}\\
Department of Mathematics and Computer Science\\ Freie Universität Berlin\\ Berlin Germany\\
	\texttt{chakri.dhanakoti@gmail.com} \\
}

\date{}

\renewcommand{\shorttitle}{\textit{arXiv}}

\hypersetup{
pdftitle={A template for the arxiv style},
pdfsubject={q-bio.NC, q-bio.QM},
pdfauthor={David S.~Hippocampus, Elias D.~Striatum},
pdfkeywords={First keyword, Second keyword, More},
}

\begin{document}

\maketitle

\begin{abstract}
The application of variational principles for analyzing problems in the physical sciences is widespread. Cantilever-like problems, where one end is fixed and the other end is free, have received less attention in terms of their stability despite their prevalence. In this article, we establish stability conditions for these problems by examining the second variation of the energy functional through the generalized Jacobi condition. This requires computing conjugate points determined by solving a set of initial value problems from the linearized equilibrium equations. We apply these conditions to investigate the nonlinear stability of intrinsically curved elastic cantilevers subject to an end load. The rod deformations are modelled using Kirchhoff rod theory. The role of intrinsic curvature in inducing complex nonlinear phenomena, such as snap-back instability, is particularly emphasized. The numerical examples highlight its dependence on the system parameters. These examples illustrate potential applications in the design of flexible soft robot arms and innovative mechanisms.
\end{abstract}


\section{Introduction}
\label{sec:s1}

Many problems in physics can be analyzed using the calculus of variations framework which has a rich history. For mechanical systems, equilibrium configurations are solutions to the force and moment balance equations, which, in this framework, can be stated as the critical points of an energy functional. The study of slender structures falls within this category and has captivated researchers since the time of Euler and Bernoulli~\cite{Matsutani2010}. Slender structures like telephone cords, ropes, cables, and
hair are found everywhere. Nonlinear rod theories have been effectively used to study the large deformations in these structures, highlighting their relevance across various fields, including Biology, Physics, and Engineering. These include small-scale domains like DNA~\cite{Manning_DNA}, bacterial locomotion~\cite{Park_Kim_Lim_2019}, nanorods~\cite{singh2022computational} as well as large-scale domains like plant tendrils~\cite{mcmillen2000}, curly hair~\cite{Miller2014}, and architecture designs~\cite{Hafner2021}.

 In recent times, soft robotics has increasingly employed slender rods to create compliant mechanisms~\cite{Rucker2010,chen2020soft}. Inspiration is often drawn from mechanisms such as octopus tentacles or elephant trunks to utilize them in real-life applications, as they are capable of increased manipulation and dexterity~\cite{Laschi2012,MajidiCarmel2014}. Advances in modern material science enabling the production of highly deformable polymers and alloys have further accelerated their development. Furthermore, the computer graphics community has been increasingly enthusiastic about using rod models to simulate realistic animations of structures like trees~\cite{Oreilly2012} and hair~\cite{romero2021physical}.

 Generally, a flexible rod can exhibit multiple equilibrium states, raising natural questions about its stability. The variational structure relates the stability of equilibria to the local minima of the functional. In the calculus of variations problems with a classical case of fixed (or Dirichlet) boundary conditions, the absence of conjugate points termed as the Jacobi condition is a well-known necessary condition for the critical points to be local minimum~\cite{bolza1904lectures, gelfand1963calculus}. Using this approach, numerous studies have examined the stability of elastic rods, covering various cases~\cite{Maddocks1984, Manning1998, Hoffman2002}. Alternative semi-analytic approaches have also been investigated on planar problems with exact analytical solutions~\cite{kuznetsov2002, Levyakov2010, Batista2015}. Nevertheless, the second-order conditions for cantilever-like problems, which have fixed-free ends, remain very obscure. For such problems, Jacobi conditions are largely restricted to one-dimensional cases~\cite{Oreilly2011, Batista2016}. These conditions have also been equivalently illustrated using the optimal control framework~\cite{Oreilly2011}. Stability studies have been conducted on planar elasticae, which are one-dimensional and have analytical solutions~\cite{Levyakov2010,Batista2015}. These conditions must be extended to higher dimensions to investigate three-dimensional elastic rod deformations. In this article, we generalize these Jacobi conditions to higher dimensions and complex boundary conditions, broadening their scope to a wider range of problems.

Cantilever structures are of particular interest because of their presence across multiple disciplines of technology. A notable example is flexible robotic arms, where one end is attached to the robot body and the other to a payload. Similarly, animations often depict cantilever structures such as trees or hair, where one end remains fixed, and the other end interacts with external forces. Many of these structures are characterized by an intrinsic curvature, resulting in complex mechanics and geometrically nonlinear behavior. In this article, we employ the developed Jacobi condition to assess the stability of the cantilever equilibria.

 The stability can be related to the associated dynamic stability of the evolutionary system and has also been examined in this context~\cite{goriely1997nonlinear,kumar2010generalized}, which includes computing eigenvalues of linearized dynamical equations. The existence of unstable equilibria suggests the possibility of snap-back instability, where an unstable equilibrium abruptly transitions to the adjacent stable equilibrium. When a naturally straight elastic rod clamped at one end and with a dead load attached to the other end is rotated using the clamp, it exhibits a snap-back instability for an appropriate combination of its length and load. This well-known catapult behavior can be harnessed in the design of soft robot arms and triggering mechanisms~\cite{Armanini2017}. Likewise, catapult behavior is possible in intrinsically curved elastic rods, i.e., one stable spatial equilibrium snaps to another spatial equilibrium, when its clamped end is rotated. However, intrinsic curvature introduces complexity, leading to non-intuitive geometrically nonlinear behavior in elastic rods. One notable consequence is the ratcheting behavior while transmitting torque-guided tubes in angioplasty~\cite{Warner1997}
or in machine shafts~\cite{VETYUKOV2023104431}. The complex shapes of curly hair~\cite{Miller2014} are also attributed to this effect. Moreover, intrinsic curvature improves the reachout of soft robotic arms compared to tip-loaded naturally straight rods~\cite{Sipos2020}. In this article, we study intrinsically curved cantilever structures with a dead load at the tip and examine their stability properties using the Jacobi condition. In particular, we assess the potentially arising snap-back instability as its clamped end is rotated. We perform a quasi-static analysis and infer the conditions of snap-back instability without resorting to discussion on their dynamics. This study offers better insight into the resulting nonlinearities and may aid in the improved design of soft robotic arms and innovative mechanisms.

 An outline of this paper is as follows: In Section~\ref{sec:s2}, we introduce the classic unconstrained calculus of variations problem with 
fixed-fixed ends and show the conditions for the local minima. This analysis is then extended to cases with fixed-free ends. In Section \ref{sec:s3}, the cantilever problem is formulated using the Kirchhoff Rod theory~\cite{antman2006nonlinear}, and equilibrium equations are derived using the Hamiltonian formulation of elastic rods~\cite{Dichmann1996}. The second variation form for these equilibria is then deduced, and a numerical strategy to compute the conjugate points is formulated. This involves evaluating a stability matrix based on the solutions to a set of initial value problems (IVP). In Section~\ref{sec:s4}, we present some examples of intrinsically curved cantilevers along with stability results, highlighting the snap-back instability. Finally, Section~\ref{sec:s5} provides a summary and discussion of the results.

\section{Calculus of Variations Problem}
\label{sec:s2}
\subsection{Fixed-Fixed Ends}

In this section, we present the standard conjugate point theory for the unconstrained calculus of variations problem and extend it to the non-classical case of fixed-free ends. The results presented here are well established in the calculus of variations literature and can be found in standard textbooks~\cite{bolza1904lectures,gelfand1963calculus}.
Let $ \boldsymbol \zeta:s\longrightarrow \mathbb{R}^{p} $ be a continuous and differentiable function in the interval $[0,l]$, where $p$ denotes the dimension of the problem. Given a continuous mapping  $\mathcal{L}:\mathbb{R}^{p} \times \mathbb{R}^{p} \times [0,l]  \rightarrow \mathbb{R}$, the standard calculus of variations problem is to minimize the functional
\begin{align}
 \label{eqn:Functional}
  J(\boldsymbol \zeta)= \int_{0}^{l} \mathcal{L}(\boldsymbol \zeta,\boldsymbol \zeta^{\prime},s)ds, \quad  \text{subject to } \boldsymbol \zeta(0)= f_{o}, \quad \boldsymbol \zeta(l)=f_{l}.
\end{align}
 The notation $()^{\prime}$ denotes the derivative with respect to $s$. The function $\mathcal{L}$ is assumed to have continuous second derivatives with respect to all its arguments and to be convex with respect to its second argument.  Critical points or equilibria $\boldsymbol \zeta_{o}$ that are expressed as solutions to the Euler-Lagrange equations:
 \begin{align}
\label{eqn:First_variation}
-\left( \frac{\partial \mathcal{L}}{\partial \boldsymbol \zeta^{\prime}}\right)^{\prime} + \frac{\partial \mathcal{L}}{\partial \boldsymbol \zeta} =\mathbf{0},
 \end{align}
 can be classified as local minima if the second variation functional evaluated as
\begin{align}
\label{eqn:second_variation}
   \delta ^{2}J(\boldsymbol \zeta_{o})[\mathbf{h}] =\frac{1}{2} \int_{0}^{l}\left( \mathbf{h}^{\prime} \cdot \mathbf{P} \mathbf{h}^{\prime} +  \mathbf{h} \cdot \mathbf{C} \mathbf{h}^{\prime}+  \mathbf{h}^{\prime}  \cdot \mathbf{C}^{T} \mathbf{h}+ \mathbf{h} \cdot  \mathbf{Q} \mathbf{h} \right)ds ,
  \end{align}
  is non-negative. The notation $\cdot$ represents the standard dot-product between two column vectors $(\mathbf{u} \cdot \mathbf{v}= \mathbf{u}^{T} \mathbf{v})$. Here, $\mathbf{P}=\mathcal{L}_{ \boldsymbol \zeta^{\prime}  \boldsymbol\zeta^{\prime} }( \boldsymbol \zeta_{o},  \boldsymbol \zeta_{o}^{\prime},s),\mathbf{C}= \mathcal{L}_{ \boldsymbol \zeta  \boldsymbol\zeta^{\prime} }( \boldsymbol \zeta_{o},  \boldsymbol \zeta_{o}^{\prime},s)$ and $\mathbf{Q}=\mathcal{L}_{ \boldsymbol \zeta  \boldsymbol \zeta}( \boldsymbol \zeta_{o},  \boldsymbol \zeta_{o}^{\prime},s)$ are $p \times p$ Hessian matrices evaluated at the critical point $\boldsymbol \zeta_{o}$. The matrices $\mathbf{P}$ and $\mathbf{Q}$ are symmetric, whereas the matrix $\mathbf{C}$ may not be. We assume that the Legendre's strengthened condition is valid throughout:
  \begin{align}
  \label{eqn:Legendre_condition}
\mathbf{P} > 0
  \end{align}
  i.e., the matrix $\mathbf{P}$ is positive definite. Here, $\mathbf{h}$ is a variation in the solutions and must satisfy the linearized boundary conditions
  \begin{align}
  \label{eqn:BCs_fixed_fixed}
\mathbf{h}(0)=\mathbf{0}= \mathbf{h}(l).
\end{align}
After integration by parts on \eqref{eqn:second_variation}, the second variation simplifies to the form 
\begin{align}
    \label{eqn:second_variation_fixed_fixed}
    \delta^{2}J(\boldsymbol \zeta_{o})[\mathbf{h}] \equiv \frac{1}{2} \int_{0}^{l} \mathcal{S} \mathbf{h} \cdot  \mathbf{h} ds, 
\end{align}
where $\mathcal{S}$ is the second-order vector self-adjoint differential \emph{Jacobi operator}:
\begin{align}
\label{eqn:Jacobi_operator}
\mathcal{S} \mathbf{h} = - \frac{d}{ds} \left( \mathbf{P} \mathbf{h}^{\prime} + \mathbf{C}^{T}\mathbf{h} \right) + \mathbf{C} \mathbf{h}^{\prime} + \mathbf{Q}\mathbf{h},
\end{align}
This system of Ordinary Differential Equations (ODEs), together with the boundary conditions~\eqref{eqn:BCs_fixed_fixed} is also referred to as the \emph{accessory boundary value problem} or \emph{Jacobi differential equations}, and its solutions are known as \emph{accessory extremals}. 

Given Legendre's strengthened condition~\eqref{eqn:Legendre_condition}, the second variation~\eqref{eqn:second_variation_fixed_fixed} is non-negative if Jacobi's condition is satisfied i.e., the critical point $\boldsymbol \zeta_{o}$ has no conjugate point. A conjugate point is defined as a point $\sigma \in [0,l]$ for which there exists a non-trivial solution satisfying:
\begin{align}
\mathcal{S} \mathbf{h}=\mathbf{0},\quad  \mathbf{h}(0)= \mathbf{h}(\sigma)=\mathbf{0}.
\label{eqn:conjugate_points}
\end{align}

\subsection{Fixed-Free Ends}
The focus of this article is to extend this well-known Jacobi condition~\eqref{eqn:conjugate_points} to cases where one end is fixed, while the other end remains free. In this case, the problem involves minimizing a functional of the form:
\begin{align}
 \label{eqn:Functional_fixed_free}
  J(\boldsymbol \zeta)= \int_{0}^{l} \mathcal{L}(\boldsymbol \zeta,\boldsymbol \zeta^{\prime},s)ds + B(\boldsymbol \zeta(l)), \qquad   \text{subject to    } \boldsymbol \zeta(0)=f_{o},
\end{align}
while the other end $s=l$ is set free. The boundary term $B$ is a continuous function of the state $\boldsymbol \zeta$ at $s=l$ and has continuous derivatives. The first order condition for the stationary points yields the Euler-Lagrange equations~\eqref{eqn:First_variation}, along with the additional natural boundary conditions at the free end $s=l$:
\begin{align}
\label{eqn:first_variation}
     \left( \frac{\partial \mathcal{L}}{\partial \boldsymbol \zeta^{\prime}}  + \frac{\partial  B}{\partial \boldsymbol \zeta} \right) \bigg|_{s=l} =\mathbf{0}.
\end{align}
The critical points correspond to the local minimum if the second variation evaluated as
  \begin{align}
  \begin{split}
          \delta ^{2}J(\zeta_{o})[\mathbf{h}]= & \int_{0}^{l} \left(- \frac{d}{ds} \left( \mathbf{P} \mathbf{h}^{\prime} +\mathbf{C}^{T}\mathbf{h} \right) + \mathbf{C} \mathbf{h}^{\prime} + \mathbf{Q}\mathbf{h} \right) \cdot \mathbf{h} ds \\
         & \qquad \qquad + \Big[\left( \mathbf{P} \mathbf{h}^{\prime} +\mathbf{C}^{T}\mathbf{h} \right)\cdot \mathbf{h} \Big]_{0}^{l}+ \mathbf{B} \mathbf{h}(l) \cdot \mathbf{h}(l) 
    \end{split}
  \end{align}
  is non-negative. The matrices $\mathbf{P}$,$\mathbf{C}$ and $\mathbf{Q}$ remain as previously defined, while the new matrix $\mathbf{B}=\frac{\partial^2 B }{\partial \boldsymbol \zeta^{2}}$ is symmetric. Although we explicitly derived most expressions for the case of fixed-fixed ends, their form remains nearly identical for the case of fixed-free ends. In this case, the variation $\mathbf{h}$ must satisfy the linearized boundary condition given by
    \begin{align}
    \label{eqn:Linearized_BCs}
\mathbf{h}(0)=\mathbf{0}=\mathbf{P} \mathbf{h}^{\prime}(l) +\mathbf{C}^{T}\mathbf{h}(l)+ \mathbf{B} \mathbf{h}(l).
\end{align}
On integrating by parts, followed by the vanishing  boundary terms, the second variation functional becomes
\begin{align}
    \label{eqn:second_variation_fixed_free}
    \delta^{2}J(\boldsymbol \zeta_{o})[\mathbf{h}] \equiv \frac{1}{2} \int_{0}^{l} \mathcal{S} \mathbf{h} \cdot  \mathbf{h} ds, 
\end{align}
where $\mathcal{S}$ is the same second order differential self-adjoint operator as define in~\eqref{eqn:Jacobi_operator}. Given Legendre's strengthened condition~\eqref{eqn:Legendre_condition}, the second-variation functional \eqref{eqn:second_variation_fixed_free} is non-negative if it satisfies the Jacobi condition, namely the absence of conjugate points. However, the definition of a conjugate point is slightly modified in this scenario. A conjugate point is defined as a point $l^{*} \in [0,l]$ for which there is a non-trivial solution to 
\begin{align}
\label{eqn:conjugate_point}
\mathcal{S}\mathbf{h}= \mathbf{0}, \quad \mathbf{h}(l^{*})= \mathbf{0}, \quad \mathbf{P} \mathbf{h}^{\prime}(l) +\mathbf{C}^{T}\mathbf{h}(l)+ \mathbf{B} \mathbf{h}(l)=\mathbf{0}.
\end{align}
Unlike the fixed-fixed case, the boundary where the natural boundary conditions (in the present case $s=l$) are given must be explicitly chosen when specifying the boundary conditions. The boundary condition at $s=l$ is accommodated by using a basis of solutions for $\mathbf{h}(l)$. In the subsequent sections, proofs supporting the revised definition of conjugate points are provided. \\

\subsubsection{Necessary condition} 
\begin{theorem}
If the matrix $\mathbf{P}$ is strictly positive definite, and the interval $[0,l]$ contains no point conjugate to $l$, then the  second variation quadratic functional $\delta^{2}J(\boldsymbol \zeta_{o})[\mathbf{h}] $ is positive for all $\mathbf{h}(s) $ satisfying the boundary conditions~\eqref{eqn:Linearized_BCs}.
\end{theorem}

\begin{proof}
This proof is a generalized version of the classical case of fixed-fixed ends presented in~\cite[page~150]{gelfand1963calculus}. Let $\mathbf{W}:s \longrightarrow \mathbb{R}^{p \times p} $ be an arbitrary differentiable symmetric matrix. Without affecting the value of the second variation integral~\eqref{eqn:second_variation}, the following term can be added
  \begin{align*}
  \begin{split}
 0= \int_{0}^{l}  \frac{d}{ds}\left( \mathbf{W} \mathbf{h} \cdot\mathbf{h}\right) ds - \big[\mathbf{W} \mathbf{h} \cdot\mathbf{h}\big] _{0}^{l}=  \int_{0}^{l}  &\mathbf{W}^{\prime} \mathbf{h} \cdot\mathbf{h} +  \mathbf{W} \mathbf{h}^{\prime} \cdot\mathbf{h} +  \mathbf{W} \mathbf{h} \cdot\mathbf{h}^{\prime} ds \\& 
 - \big( \mathbf{W}(l) \mathbf{h}(l) \cdot\mathbf{h}(l) -  \mathbf{W}(0) \mathbf{h}(0) \cdot\mathbf{h}(0) \big).
\end{split}
\end{align*}
The matrix function $\mathbf{W}(s)$ is chosen such that the boundary terms vanish. For the current case with fixed-free ends, we have $\mathbf{h}(0)=\mathbf{0}$ at the fixed end, and additionally we require that $\mathbf{h}(l) \cdot \left(\mathbf{W}(l) \mathbf{h}(l) + \mathbf{B} \mathbf{h}(l) \right)=0$ at the free end. The latter condition holds for any non-trivial $\mathbf{h}(l)$ when $\mathbf{W}(l) + \mathbf{B}$ is a zero matrix of order $p$, denoted by $\mathbf{O}$. This condition also holds if $\mathbf{W}(l)  + \mathbf{B} $ is a skew-symmetric matrix, but it contradicts the symmetric matrix assumption of $\mathbf{W}(s)$ and, therefore, is disregarded. Then, the integral~\eqref{eqn:second_variation} becomes 
  \begin{align}
   \label{eqn:second_variation_integral}
      \delta^{2} J(\boldsymbol \zeta_{o})[\mathbf{h}] &=  \int_{0}^{l} \mathbf{P}  \mathbf{h}^{\prime} \cdot \mathbf{h}^{\prime} + \left(\mathbf{C} + \mathbf{W} \right) \mathbf{h} \cdot \mathbf{h}^{\prime} +  \left(\mathbf{C} + \mathbf{W} \right)^{T}  \mathbf{h}^{\prime}  \cdot \mathbf{h} + \left( \mathbf{Q}  + \mathbf{W}^{\prime}  \right)\mathbf{h} \cdot \mathbf{h}  ds.
  \end{align}
The integrand can be expressed as a perfect square of the form 
    \begin{align*}
    \delta^{2} J(\boldsymbol \zeta_{o})[\mathbf{h}]=
      \int_{0}^{l} \left( \mathbf{P} ^{1/2} \mathbf{h}^{\prime} + \left(\mathbf{Q} + \mathbf{W}^{\prime}\right)^ {1/2} \mathbf{h}  \right)\cdot  \left( \mathbf{P} ^{1/2} \mathbf{h}^{\prime} + \left(\mathbf{Q} + \mathbf{W}^{\prime }   \right)^{1/2} \mathbf{h} \right)ds,
  \end{align*}
  if the matrix $\mathbf{W}(s)$ is chosen to be the solution of
 \begin{align}
\mathbf{Q} + \mathbf{W}^{\prime} = \left( \mathbf{C} + \mathbf{W} \right) \mathbf{P}^{-1}\left( \mathbf{C} ^{T}+ \mathbf{W} \right) .
     \label{eqn:Matrix_Ricatti_eqn}
   \end{align}
Since $\mathbf{P}$ is assumed to be a positive-definite symmetric matrix, its square root $\mathbf{P}^{1/2}$ exists and is also positive-definite. Moreover, its inverse $\mathbf{P}^{-1/2}$ exists. The expression~\eqref{eqn:Matrix_Ricatti_eqn} is called as \emph{Matrix Ricatti equation}, and the second variation integral~\eqref{eqn:second_variation_integral} takes the form
 \begin{align}
    \delta^{2}J(\boldsymbol \zeta_{o})[\mathbf{h}] = \int_{0}^{l} \mathbf{P} \left( \mathbf{h}^{\prime} + \mathbf{P}^{-1/2}\left(\mathbf{Q} + \mathbf{W} \right) \mathbf{Q}^{-1/2}\mathbf{h}\right) \cdot  \left( \mathbf{h}^{\prime} + \mathbf{P}^{-1/2}\left(\mathbf{Q} + \mathbf{W} \right) \mathbf{Q}^{-1/2} \mathbf{h}\right) ds,
  \end{align} where the integrand is a perfect square and is always non-negative. The second variation is zero only when the expression
  \[
  \mathbf{P}^{1/2} \mathbf{h}^{\prime} + \left(\mathbf{Q}+ \mathbf{W} \right)\mathbf{Q}^{-1/2}\mathbf{h}
  \]
  vanishes, and this is only possible for the trivial solution $\mathbf{h}(s)=\mathbf{0}$. If the \emph{Matrix Ricatti equation}~\eqref{eqn:Matrix_Ricatti_eqn} has a continuous solution $\mathbf{W}(s)$ defined over the interval $[0,l]$, then the second variation is positive-definite. Substituting 
  \begin{align}
 \label{eqn:theorem_inverse}
  \mathbf{P} \mathbf{U}^{\prime} + \mathbf{C}^{T} \mathbf{U} + \mathbf{W} \mathbf{U}    =\mathbf{O},
    \end{align} 
  where $\mathbf{U}$ is an unknown matrix results in
  \begin{align}
  \label{eq:Euler_LinearizedEL}
      -\frac{d}{ds}\left(\mathbf{P}\mathbf{U}^{\prime} + \mathbf{C}^{T} \mathbf{U}\right) +\left( \mathbf{C} \mathbf{U}^{\prime} + \mathbf{Q} \mathbf{U} \right)=\mathbf{O} ,
  \end{align}
  which is the matrix form of the \emph{Jacobi operator} $\mathcal{S}$. Now, consider the free end, where the matrix $\mathbf{W}(l) + \mathbf{B}$ is chosen to be a zero matrix, and the relation~\eqref{eqn:theorem_inverse} yields the boundary condition
  \begin{align}
  \label{eqn:Ricatti_BCs}
     \mathbf{P}\mathbf{U}^{\prime}(l)+\mathbf{C}^{T} \mathbf{U}(l) + \mathbf{B}\mathbf{U}(l)= \mathbf{O}. 
  \end{align}
This is a matrix form of the linearized natural boundary conditions at the boundary $s=l$. The columns in matrix $\mathbf{U}$ can be interpreted as the basis of the variations $\mathbf{h}$. If $[0,l]$ contains no point \emph{conjugate} to $l$, then~\eqref{eq:Euler_LinearizedEL} has a solution $\mathbf{U}(s)$ which is non-singular in $[0,l]$. Therefore, the \emph{Matrix Ricatti equation} \eqref{eqn:Matrix_Ricatti_eqn} has a solution given by~\eqref{eqn:theorem_inverse}, satisfying~\eqref{eqn:Ricatti_BCs}. Thus, there exists a matrix $\mathbf{W}(s)$ that transforms the integrand to a perfect square, producing a non-negative second variation $\delta ^{2}J(\boldsymbol \zeta_{o})[\mathbf{h}]$.
 \end{proof}

\subsubsection{Sufficient condition} 
\begin{theorem}
If the matrix $\mathbf{P}$ is positive-definite and the interval $[0,l]$ contains a point conjugate to $l$, then the second variation quadratic functional $\delta^{2}J(\boldsymbol \zeta_{o})[\mathbf{h}] $ is not positive for all $\mathbf{h}$ satisfying the boundary conditions~\eqref{eqn:Linearized_BCs}.
\end{theorem}

\begin{proof}
Suppose there exists a point $s=l^{*}$ \emph{conjugate} to $s=l$ in $0<s<l$. Consequently, there exists a non-null accessory extremal $\mathbf{h}(s)$ satisfying $\mathbf{h}(l^{*})=\mathbf{0} $ and $\mathbf{P}\mathbf{h}^{\prime}(l) + \mathbf{C}^{T}\mathbf{h}(l) + \mathbf{B} \mathbf{h}(l)=\mathbf{0} $. Let $\boldsymbol \gamma(s)$ be a continuous arc defined as
\begin{align*}
\boldsymbol \gamma(s)=
\begin{cases}
\mathbf{0},  &0<s<l^{*},\\
 \mathbf{h}(s),&l^{*}<s<l,
\end{cases}
\end{align*}
and is depicted in Figure~\ref{fig:Broken_Extremal}. The second variation $\delta^{2}J(\boldsymbol \zeta_{o})$ along the arc $\boldsymbol \gamma$ is given by
\begin{align*}
\delta^{2}J(\boldsymbol \zeta_{o})[\boldsymbol \gamma]&=\frac{1}{2} \int_{0}^{l} \mathcal{S} \boldsymbol \gamma(s) \cdot \boldsymbol \gamma(s)  ds=\frac{1}{2}\int_{l^{*}}^{l} \mathcal{S}\mathbf{h}(s) \cdot \mathbf{h}(s)  ds, \\
&=\frac{1}{2}\left[\left(\mathbf{P}(s)\mathbf{h}^{\prime}(s)+ \mathbf{C}^{T}(s) \mathbf{h}(s)\right) \cdot \mathbf{h}(s)  \right]_{l}^{l^{*}} + \mathbf{B} \mathbf{h}(l)\cdot \mathbf{h}(l) = 0.
\end{align*}

\begin{figure}[tbh]
    \centering
    \includegraphics[width=0.8\textwidth]{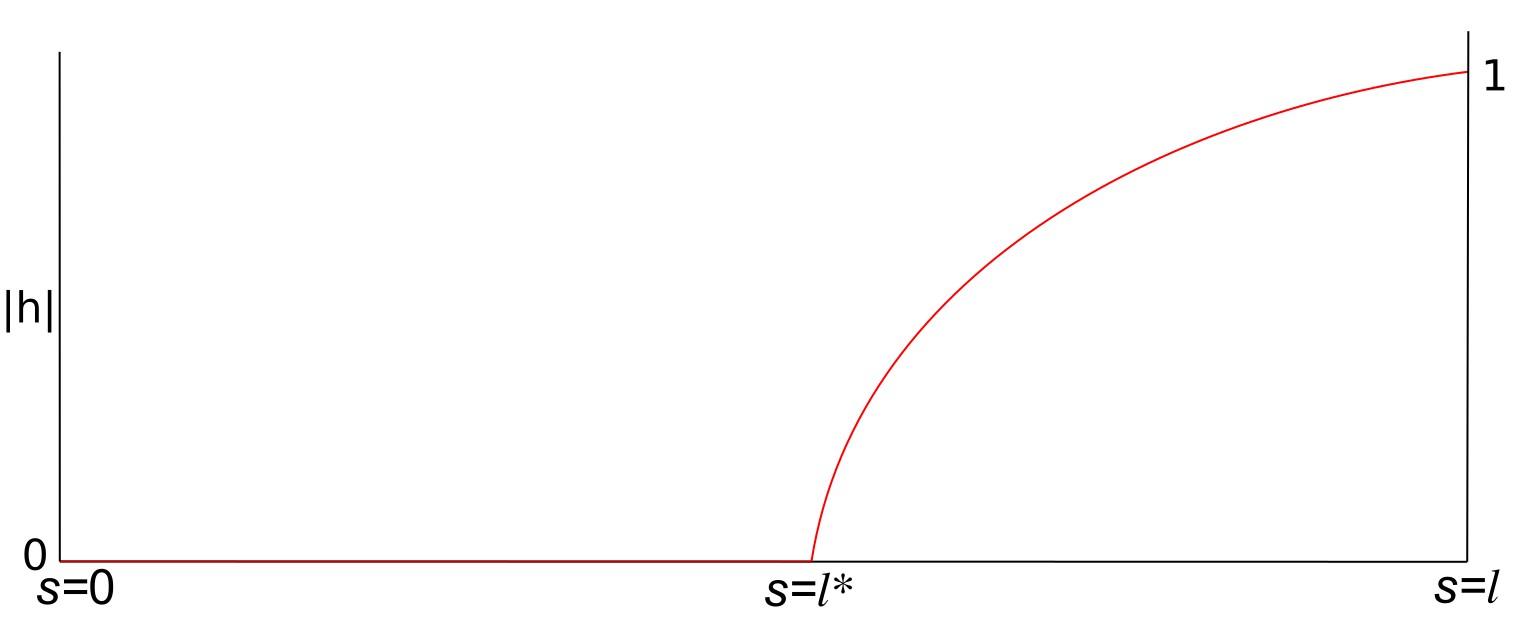}
    \caption{The broken accessory extremal $ \gamma(s) $ satisfying the boundary conditions~\eqref{eqn:Linearized_BCs} }
    \label{fig:Broken_Extremal}
\end{figure}

However, this arc $\boldsymbol \gamma(s)$ has a \emph{corner point} at $s=l^{*}$, as $\mathbf{h}^{\prime}(l^{*}-) \neq \mathbf{h}^{\prime}(l^{*}+)$. If the matrix $\mathbf{P}(s)$ is positive definite, then by \emph{Weierstrass-Erdmann conditions}~\cite{bolza1904lectures}, the arc $\gamma$ with a \emph{corner point} cannot be the local minimizer. But, the second variation functional  $\delta^{2}J$ is zero along the present \emph{broken extremal} $\gamma(s)$. Therefore, there must exist another arc $\mathbf{h}(s)$ satisfying the boundary conditions~\eqref{eqn:Linearized_BCs} which is a local minimizer and further reduces the second variation $\delta^{2}J<0$, thereby proving the theorem.

\end{proof}
Therefore, these two conditions prove that the critical points correspond to local minima if and only if they do not have any conjugate points; otherwise, they are not local minima.

\section{Application to Cantilever structures}
\label{sec:s3}
The main motivation for developing the conjugate point stability test is to evaluate the stability of tip-loaded soft cantilever arms. The deformations in the elastic rods are modelled using the standard Kirchhoff Rod theory~\cite{antman2006nonlinear}, and  we describe this theory in this section. Euler parameters are employed in our model to describe rotations, and most of the notation adopted here follows~\cite{Dichmann1996}.

\subsection{Kinematics and Equilibrium}

The elastic rod configuration is modelled as an orientable curve in 3D space, represented by a centreline $\mathbf{r}:\left[0,l\right] \ni s  \longrightarrow \mathbb{R}^{3} $ and an orientation frame spanned by the orthonormal unit vectors called directors $\mathbf{d}_{i}:\left[0,l\right] \ni s \longrightarrow \mathbb{R}^{3}, i=1,2,3$. The independent variable $s$ is the arclength of the undeformed configuration. The directors $\mathbf{d}_{1}$ and $\mathbf{d}_{2}$ span along the cross-section of the rod and satisfies $\mathbf{d}_{3}=\mathbf{d}_{1} \times \mathbf{d}_{2}$. Here, $\times$ denotes the standard cross-product between two vectors. The rate of change of the director frame $\{ \mathbf{d}_{1},\mathbf{d}_{2},\mathbf{d}_{3} \}$ with respect to the arclength $s$ is characterized using the Darboux vector $\mathbf{u} \in \mathbb{R}^{3}$ as
\begin{equation}
\begin{split}
\mathbf{d}_{i}^{\prime}\left( s \right)&= \mathbf{u}\left( s \right) \times \mathbf{d}_{i}\left( s \right), \qquad  i=1,2,3.
\end{split}
\end{equation}
The components along the local directors $u_{1}\equiv \mathbf{u} \cdot \mathbf{d}_{1}$, $u_{2}\equiv \mathbf{u} \cdot \mathbf{d}_{2}$  correspond to bending strains, while the component $u_{3} \equiv \mathbf{u} \cdot \mathbf{d}_{3}$ corresponds to twisting strain.  We restrict ourselves to the case of inextensible and unshearable rods, where the tangent $\mathbf{d}_{3}$ coincides with the tangent to the centreline $\mathbf{r}(s)$
\begin{equation}
\label{eqn:constrain1}
\mathbf{r}^{\prime}(s)= \mathbf{d}_{3}\left(s \right),
\end{equation}
where $()^{\prime}$ denotes the differentiation with respect to arclength $s$. We denote these local strain components using a triad $\mathtt{u}(s)=(u_{1}(s),u_{2}(s),u_{3}(s))$. We consider the unstressed configuration or lowest energy configuration to be the reference state. Let $\mathtt{\hat{u}}(s)=(\hat{u}_{1}(s),\hat{u}_{2}(s),\hat{u}_{3}(s))$ be  the triplet of strain components in its unstressed configuration. A constant $\mathtt{\hat{u}}(s)=(0,0,0)$ represents a intrinsically straight rod, while a constant intrinsic curvature of the form $(\hat{u}_{1},\hat{u}_{2},0)$, with both of $\hat{u}_{1},\hat{u}_{2}$ not zero, represents a circular arc. A constant non-zero $\hat{u}$ in the later case, transforms it into a helix shape. A constant $\mathtt{\hat{u}}(s)=(0,0,0)$ represents a intrinsically straight rod, while a constant intrinsic curvature of the form $(\hat{u}_{1},\hat{u}_{2},0)$, with both $\hat{u}_{1}$,$\hat{u}_{2}$ not zero, represents a circular arc. A constant non-zero $\hat{u}$ in the latter case, transforms it into a helix shape. The local orientation frame $\{ \mathbf{d}_{1},\mathbf{d}_{2},\mathbf{d}_{3} \}$ is connected to the fixed laboratory  frame  $\{ \mathbf{e}_{1},\mathbf{e}_{2},\mathbf{e}_{3} \}$  through an $SO(3)$ matrix, which is commonly parameterized using three Euler Angles. However, the three-dimensional representation of Euler angles does not completely represent the $SO(3)$-space globally and has singular directions. The Euler parameters (or a Quaternions) $\textbf{q} \in  \mathbb{R}^{4}$~\cite{goldstein2011classical, shuster1993survey} with the property  $ \|\mathbf{q} \|=1$, provides the global representation of $SO(3)$ space and mitigate the singularities. In addition, Euler parameter representation uses quadratic functions which are computationally quicker compared to the trigonometric functions used in Euler angle representation. Any spatial orientation of a rigid body can be uniquely described as a simple rotation $\alpha$ about an arbitrary axis $a_{1} \mathbf{e}_{1} + a_{2}\mathbf{e}_{2} + a_{3} \mathbf{e}_{3}$. Then, the corresponding Euler parameters are defined as
\begin{align}
q_{i}= a_{i} \sin \frac{\alpha}{2} \quad i=1,2,3, \qquad
q_{4}=\cos \frac{\alpha}{2}.
\end{align}
The directors $\mathbf{d}_{i} \in \mathbb{R}^{3}, i=1,2,3$ with respect to the fixed coordinate system $\{\mathbf{e}_{1},\mathbf{e}_{2},\mathbf{e}_{3} \}$ can be written in terms of Euler parameters $\mathbf{q}$  as
\begin{align}
\begin{split}
&\mathbf{d}_{1} = \frac{1}{\left|\mathbf{q} \right| ^{2}} 
\begin{bmatrix}
q^{2}_{1}- q^{2}_{2}- q^{2}_{3} +q^{2}_{4} \\
2 \left( q_{1} q_{2} + q_{3} q_{4} \right) \\
2 \left( q_{1} q_{3} - q_{2} q_{4} \right)
\end{bmatrix},  \qquad
\mathbf{d}_{2}= \frac{1}{\left|\mathbf{q} \right| ^{2}} 
\begin{bmatrix}
2 \left( q_{1} q_{2} - q_{3} q_{4} \right) \\
-q^{2}_{1}+ q^{2}_{2}- q^{2}_{3} +q^{2}_{4} \\
2 \left( q_{2} q_{3} + q_{1} q_{4} \right),
\end{bmatrix},\\ 
& \hspace{2cm} \mathbf{d}_{3} = \frac{1}{\left|\mathbf{q} \right| ^{2}} 
\begin{bmatrix}
2 \left( q_{1} q_{3} + q_{2} q_{4} \right) \\
2 \left( q_{2} q_{3} - q_{1} q_{4} \right)\\
-q^{2}_{1}- q^{2}_{2}+ q^{2}_{3} +q^{2}_{4} 
\end{bmatrix}.
\end{split}
\end{align}
Similarly, the strain components $ \mathbf{u}(s) \cdot \mathbf{d}_{j}(s) \equiv $ $u_{j}(s)$ in terms of Euler parameters and their derivatives~\cite{Dichmann1996} are given by
\begin{equation}
{u}_{j}(s)=\frac{2}{\mathbf{\left|q \right|}^2}\mathbf{B_{j}} \mathbf{q}\cdot  \mathbf{q}^{\prime}, \qquad j =1,2,3,
\label{eqn:strain_quaternions}
\end{equation}
where $\mathbf{B}_{j}$, $j=1,2,3$ are $ 4 \times 4 $ skew symmetric matrices given by
\begin{equation}
\label{eqn:Bmatrices}
\mathbf{B}_{1}=
\begin{bmatrix}
0 & 0& 0& 1 \\
0 & 0& 1& 0 \\
0 & -1& 0& 0 \\
-1 & 0& 0& 0 
\end{bmatrix}, \quad
\mathbf{B}_{2}=
\begin{bmatrix}
0 & 0& -1& 0 \\
0 & 0& 0& 1\\
1 & 0& 0& 0 \\
0 & -1& 0& 0 
\end{bmatrix},\quad
\mathbf{B}_{3}=
\begin{bmatrix}
0 & 1& 0& 0 \\
-1 & 0& 0& 0 \\
0 & 0& 0& 1 \\
0 & 0& -1& 0 
\end{bmatrix}.
\end{equation}
These matrices map $\mathbf{q} \in \mathbb{R}^{4} $ to vectors that are orthogonal to each other as well as orthogonal to $\mathbf{q}(s)$.

The stresses acting along the cross-section of the rod can be averaged out to yield internal force $\mathbf{n} \in \mathbb{R}^{3}$ and moment $\mathbf{m} \in \mathbb{R}^{3} $. The components $m_{i}(s) \equiv\mathbf{m}(s) \cdot \mathbf{d}_{i}(s), i=1,2$ are the bending moments and the component $m_{3}(s)\equiv\mathbf{m}(s) \cdot \mathbf{d}_{3}(s)$ is the twisting moment in the rod. We use $\mathtt{m}$ to denote the triad of these components $\mathtt{m}=\left(m_{1},m_{2},m_{3}\right)$. We consider the rods that satisfy Hyperelastic constitutive law. A convex strain energy density function $ W:\{\mathbf{w};s \}\rightarrow \mathbb{R}^{+}$, $\mathbf{w}=(w_{1},w_{2},w_{3})$ exists such that $\frac{\partial W(0,s)}{\partial w_{i}}=0,i=1,2,3,$ $\forall s$, and the moment components are given by
\begin{equation}
m_{i}(s) = \frac{\partial }{\partial w_{i}}W(w_{i},s),     \quad  i=1,2,3,
\end{equation}
where the shifted strain argument $w_{i} \equiv u_{i}-\hat{u}_{i}$ describe the strain from intrinsic shape $\hat{u}_{i}$. In the present work, we restrict ourselves to a simple linearly elastic constitutive model where the strain energy density function is given by
\begin{equation}
\label{eqn:strinenergy}
W \left(u_{i} -\hat{u}_{i};s  \right) =\sum_{i=1}^{3} \frac{1}{2}
K_{i}(s) \left(u_{i} (s)-\hat{u}_{i} (s)\right)^{2}, 
\end{equation}
the moment components are given by
\begin{align}
    m_{i}= \frac{\partial W}{\partial \mathtt{u}_{i}}=K_{i}(s) \left( u_{i} (s)-\hat{u}_{i} (s)\right), \hspace{1cm} i=1,2,3.
\end{align}
Here, $K_{i} :s\rightarrow \mathbb{R} $ for $i=1,2$ is called bending stiffness  or flexural  rigidity, and $K_{3} :s\rightarrow \mathbb{R} $ is called torsional stiffness of the rod. Generally, a straight elastic rod with a circular cross-section follows a transversely isotropic law~\cite{antman2006nonlinear}, leading to $K_{1}=K_{2}$. We consider a similar transversely isotropic law in our intrinsically curved elastic rods. In structural mechanics, the bending stiffnesses $K_{i}$ $(i=1,2)$ at any cross-section $s$ along the rod is given by $EI$, while the torsional stiffness is given by $ EI/(1+\nu)$. Here $E$ is Young's modulus of the material, which characterizes its strength, $I$ is the second moment of area of the cross-section of the rod, which depends on its geometry, and $\nu$ is the Poisson's ratio of the material. Its value ranges from $0$ to $0.5$, with $\nu = 0.5$ corresponding to an incompressible material. For a circular cross-section rod, $I = \frac{\pi}{4} r^{4}$, where $r$ is its cross-sectional radius.
\begin{figure}[t]
    \centering
    \includegraphics[width= 0.7\textwidth]{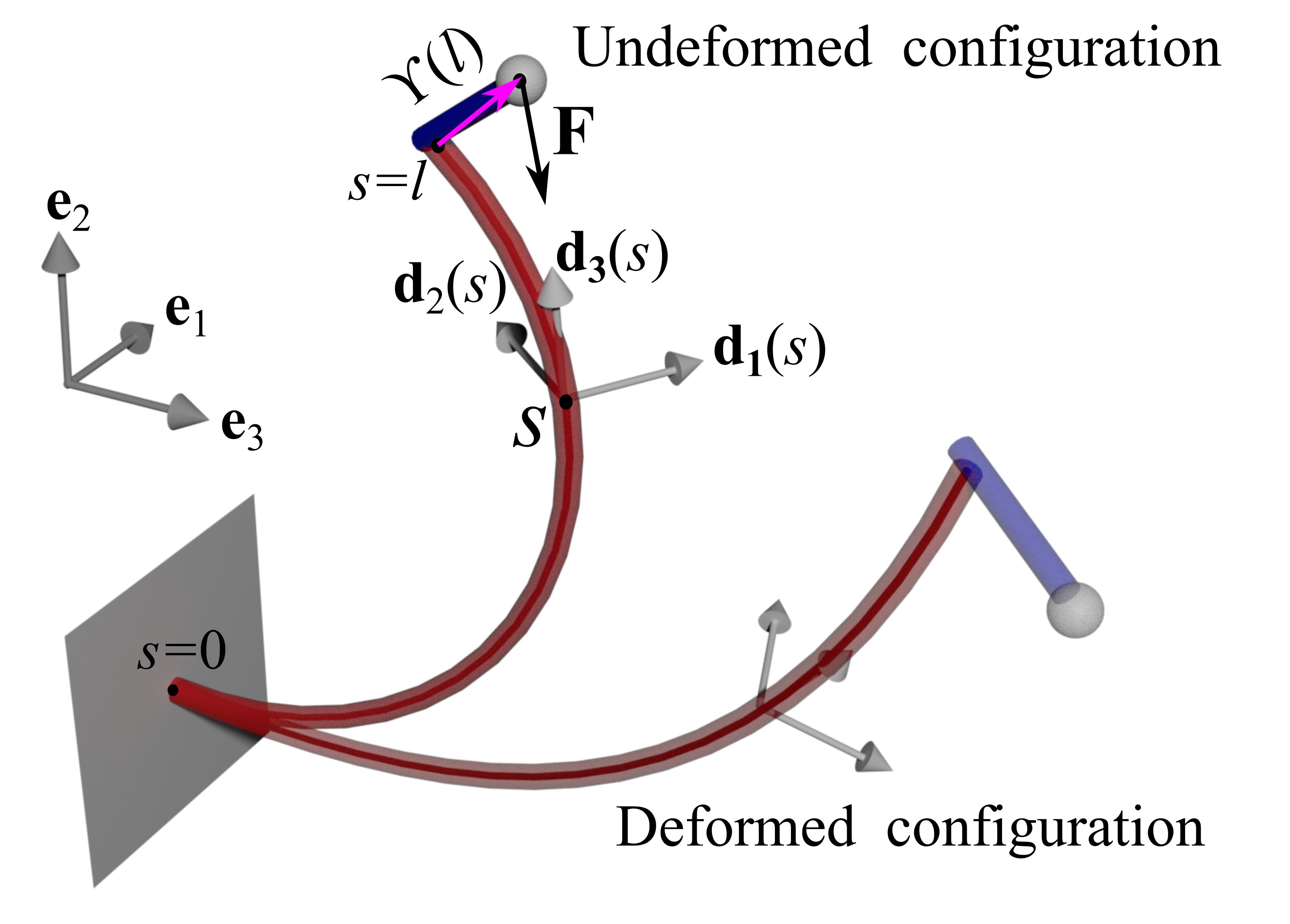}
    \caption{Schematic showing  an elastic rod with an external tip load acting through a massless rigid arm (in blue). The arm is assumed to be fixed to the tip of the elastic rod.}
    \label{fig:figure2}
\end{figure}

We consider a problem where a massless elastic rod is clamped at one end, with a dead payload applied at the other end as shown in Figure~\ref{fig:figure2}. The payload is rigidly attached to the tip $s=l$, so its spatial orientation depends on the tip's orientation frame  $\{\mathbf{d}_{1}(l),\mathbf{d}_{2}(l),\mathbf{d}_{3}(l) \}$, which is a function of $\mathbf{q}(l)$. The lever arm connecting the point of attachment to the point of application of force in the fixed frame is given by
\begin{align}
    \Upsilon(\mathbf{q}(l))\equiv \Delta_{1} \mathbf{d}_{1}(l)+ \Delta_{2}\mathbf{d}_{2}(l)+ \Delta_{3}\mathbf{d}_{3}(l).
\end{align}
where $\Delta_{1},\Delta_{2},\Delta_{3}$ are constant scalars representing the components of the arm in the tip's frame, and $\boldsymbol{\Delta}=(\Delta_{1},\Delta_{2},\Delta_{3})$ represents this triad. The payload exerts a constant force $\mathbf{F} \in \mathbb{R}^{3}$ and a moment $\boldsymbol \Upsilon(\mathbf{q}(l))\times  \mathbf{F} \in \mathbb{R}^{3}$ at the tip $s=l$. Then, the stored potential energy due to this applied tip load $\mathbf{F}$ is given by
\begin{align*}
E_{l}=-\mathbf{F} \cdot \left( \mathbf{r}(l)  +  \Upsilon(\mathbf{q}(l)) \right).
\end{align*}
The rod equilibria are the constrained critical points of the strain energy functional of the total energy functional, which is the sum of elastic strain energy and stored potential energy
\begin{align}
\int_{0}^{l}W \left(\frac{2}{\mathbf{\left|q \right|}^2}\mathbf{B}_{j} \mathbf{q}\cdot  \mathbf{q}^{\prime} - \hat{u}_{j},s  \right)  ds - \mathbf{F} \cdot \left( \mathbf{r}(l)  +  \Upsilon (\mathbf{q}(l)) \right),
\label{eqn:Energy_functional}
\end{align}
subject to pointwise constraints
\begin{align}
 \mathbf{r}^{\prime} - \mathbf{d}_{3}=\mathbf{0},\quad \mathbf{q} \cdot \mathbf{q}^{\prime}=0.
\end{align}
By definition, the Euler parameters have a unit length $\mathbf{q} \cdot \mathbf{q}=1$. Instead of enforcing it, we choose to impose this constraint on its derivative $ \mathbf{q} \cdot \mathbf{q}^{\prime}=0$. This equation, together with the boundary conditions $\mathbf{q}(0) \cdot \mathbf{q}(0)=1$ is equivalent to $\mathbf{q}(s) \cdot \mathbf{q}(s)=1$. This constrained variational problem is formulated as an unconstrained variational problem using the Lagrange multipliers $\boldsymbol \lambda: s\rightarrow   \mathbb{R}^{3}$, $\eta: s\rightarrow  \mathbb{R}$. The functional represented by
\begin{equation}
\begin{split}
\label{eqn:augmented_Lagrangian}
 J  &= \int_{0}^{l} L(\mathbf{r},\mathbf{r}^{\prime},\mathbf{q},\mathbf{q}^{\prime}; s) ds \\
    &= \int_{0}^{l} W \left(\frac{2}{\mathbf{\left|q \right|}^2}\mathbf{B}_{j} \mathbf{q}\cdot  \mathbf{q}^{\prime} - \hat{u}_{j},s  \right)  +  \boldsymbol \lambda \cdot \left( \mathbf{r}^{\prime} - \mathbf{d}_{3} \right) + \eta \mathbf{q} \cdot \mathbf{q}^{\prime} ds  
    - \mathbf{F} \cdot \left( \mathbf{r}(l)  +  \Upsilon(\mathbf{q}(l)) \right).
\end{split}
\end{equation}
is stationary at the equilibrium. We refer to the integrand $L$ as the Augmented Lagrangian. The Hamiltonian form of the equilibria~\cite{Dichmann1996} is adapted here, as they offer more simplicity in terms of analysis and numerics. The phase variables in this Hamiltonian formulation are the states $\mathbf{r}, \mathbf{q}$ and their corresponding conjugate momenta $\mathbf{n}$, $\boldsymbol \mu$. The internal force $\mathbf{n}$ and the impetus $\boldsymbol \mu$ are defined using the Augmented Lagrangian $L$ as

\begin{align}
\label{eqn:momentum1}
\mathbf{n}\equiv \frac{\partial L}{\partial \mathbf{r}^{\prime}}= \boldsymbol \lambda,  \qquad
 \boldsymbol \mu \equiv \frac{\partial L}{\partial \mathbf{q}^{\prime}}= 2W_{u_{i}}   \frac{ \mathbf{B}_{i}\mathbf{q}}{ \mathbf{\left|q \right|}^2} + \eta \mathbf{q}, \quad i=1,2,3.
\end{align}
The dot product of $\boldsymbol \mu$ with $\mathbf{B}_{i}\mathbf{q}/2$ fetches the components of the internal moment $m_{i}$ in terms of Hamiltonian variables
\begin{align}
   \frac{\boldsymbol \mu \cdot \mathbf{B}_{i}\mathbf{q}}{2} & =  W_{u_{i}} \equiv m_{i} , \quad \text{for} \quad i=1,2,3.
   \label{eqn:moment_components}
\end{align}
On the other hand, the dot product of $\boldsymbol \mu$ with $\mathbf{q}$ gives the Lagrange multiplier $\eta$. The Lagrange multiplier associated with the unit constraint $\mathbf{q} \cdot \mathbf{q}=1 $ can be shown to be twice the bending energy~\cite{Oreilly1999}. Moreover, $\eta$, the Lagrange multiplier associated with $\mathbf{q} \cdot \mathbf{q}^{\prime}=0$, is negative of its anti-derivative. The Hamiltonian of the system after taking the Legendre transformation of $L$ appears as
\begin{align}
\label{eqn:Hamiltonian}
     H \left(\mathbf{r},\mathbf{q},\mathbf{n},\boldsymbol \mu;s\right) = \mathbf{n}\cdot \mathbf{d}_{3}+ \sum_{i=1}^{3} m_{i} \hat{u}_{i}  + \frac{m_{i}^2}{2K_{i}}.
\end{align}
Then, the canonical form of the Hamiltonian system of equations governing the equilibria is given by
\begin{subequations}
\label{eqn:Ham_rod_equilibria}
\begin{align}
    \label{eqn:Ham_rod_equilibria_a}
    \mathbf{r}^{\prime}(s) &=\frac{\partial H}{\partial \mathbf{n}}= \mathbf{d}_{3}, \\
    \label{eqn:Ham_rod_equilibria_b}
    \mathbf{n}^{\prime}(s) &=-\frac{\partial H}{\partial \mathbf{r}}= \mathbf{0},\\
    \label{eqn:Ham_rod_equilibria_c}
  	\qquad \mathbf{q}^{\prime}(s)& =\frac{\partial H}{\partial \boldsymbol \mu}    =  \sum_{j=1}^{3}\left( K_{j}^{-1}m_{j} + \hat{u}_{j} \right)\frac{1}{2}\mathbf{B}_{j}\mathbf{q}, \\
   \label{eqn:Ham_rod_equilibria_d}
		\qquad \boldsymbol  \mu ^{\prime}(s) & =- \frac{\partial H}{\partial \mathbf{q}}=  \sum_{j=1}^{3} \left( K_{j}^{-1} m_{j} + \hat{u}_{j} \right)\frac{1}{2} \mathbf{B}_{j}  \boldsymbol \mu   -\frac{\partial \mathbf{d}_{3}}{\partial \mathbf{q}} ^{T}\mathbf{n},
\end{align}
\end{subequations}
where the derivative $\frac{\partial \mathbf{d}_{3}}{\partial \mathbf{q}}$ is 
\begin{align*}
    \frac{\partial \mathbf{d}_{3}}{\partial \mathbf{q}}= 
      2\begin{bmatrix}
    q_{3} & q_{4} & q_{1} & q_{2},\\
    -q_{4} & q_{3} & q_{2} & -q_{1},\\
    -q_{1} & -q_{2} & q_{3} & q_{4}
    \end{bmatrix},
\end{align*}
and the component $m_{i}$ are written in terms of phase variables $\boldsymbol \mu$ and $\mathbf{q}$ using the relation~\eqref{eqn:moment_components}. 

The equations~\eqref{eqn:Ham_rod_equilibria_c},~\eqref{eqn:Ham_rod_equilibria_d} are equivalent to the balance of moments projected on the hyperspace spanned by vectors $ \mathbf{B}_{1} \mathbf{q}$, $ \mathbf{B}_{2} \mathbf{q}$ and $ \mathbf{B}_{3} \mathbf{q}$. (For more details see~\cite[page 33-35]{spcdhanakoti}). These ODEs are subjected to fixed boundary conditions at $s=0$
\begin{align}
\label{eqn:Rod_bcs_Dir}
\mathbf{r}(0)&= [0,0,0]^{T}, \qquad \mathbf{q}(0)=\mathbf{q}_{o},
\end{align}
and natural boundary conditions at the other end $s=l$
\begin{subequations}
\label{eqn:Rod_bcs_Nat}
\begin{align}
\label{eqn:Rod_bcs_Nat_a}
\mathbf{n}(l)&- \mathbf{F}= \mathbf{0}, \\
\label{eqn:Rod_bcs_Nat_b}
 m_{i}(l)&- \left( \boldsymbol \Upsilon (\mathbf{q}(l)) \times \mathbf{F} \right) \cdot \mathbf{d}_{i}(\mathbf{q}(l)) =0, \hspace{1 cm} i=1,2,3,\\
 \label{eqn:Rod_bcs_Nat_c}
 \boldsymbol \mu(l) & \cdot \mathbf{q}(l)+ 2 \mathbf{r}(l) \cdot \mathbf{n}(l) =0.
\end{align}
\end{subequations}
The quantity $\mathbf{q}_{o} \in \mathbb{R}^{4}$ defines the prescribed orientation of the director frame at the fixed end. The condition~\eqref{eqn:Rod_bcs_Nat_b} results from projection of natural boundary conditions of $\boldsymbol \mu(l)$ onto $\{ \mathbf{B}_{1} \mathbf{q}, \mathbf{B}_{2} \mathbf{q}, \mathbf{B}_{3} \mathbf{q} \}$-space and the remaining~\eqref{eqn:Rod_bcs_Nat_c} results from their projection onto $\mathbf{q}$. The last expression~\eqref{eqn:Rod_bcs_Nat_c} can be set to any value as it specifies the Lagrange multiplier $ \eta$ with a boundary condition and restricts its gauge freedom~\cite{LiMaddocks1996}. 

\subsection{Stability Analysis}
\label{sec:Elastic_rods_Jacobi_condition}
The equilibria $\boldsymbol \zeta_{o}$ obtained as the solutions to~\eqref{eqn:Ham_rod_equilibria} with the boundary conditions~\eqref{eqn:Rod_bcs_Dir},~\eqref{eqn:Rod_bcs_Nat} must satisfy the Legendre's strengthened condition~\eqref{eqn:Legendre_condition} along with the Jacobi condition to represent a local minima of the functional~\eqref{eqn:Energy_functional}. The variable $\mathbf{r}$ has no explicit contribution in the elastic strain energy $W(\mathtt{u}-\mathtt{\hat{u}},s)$; it appears only through boundary conditions and the constraint $\mathbf{r}^{\prime}=\mathbf{d}_{3}$. We decouple $\mathbf{r}(s)$ and its conjugate momentum $\mathbf{n}(s)$ from the variational problem for stability analysis by directly substituting the relations in the functional to yield
\begin{align}
\int_{0}^{l}L ds = \int_{0}^{l}\mathbf{W}\left(\frac{2}{\mathbf{\left|q \right|}^2}\mathbf{B_{j}} \mathbf{q}\cdot  \mathbf{q}^{\prime} - \hat{u}_{j},s  \right) - \mathbf{F} \cdot  \mathbf{d}_{3} ds -   \mathbf{F} \cdot \Upsilon (\mathbf{q}(l)).
\label{eqn:Energy_functional_}
\end{align}
The variations in $\mathbf{q}$ and $\boldsymbol \mu$ are represented using $\delta \mathbf{q} $ and $\delta \boldsymbol \mu $ respectively. 
In this analysis, the variations $\delta \mathbf{q}$ must be restricted to satisfy the unit-norm constraint $\mathbf{q} \cdot \mathbf{q} = 1 $. Consequently, $\delta \mathbf{q}$ satisfies $\mathbf{q} \cdot  \delta \mathbf{q}=0 $ or equivalently $\delta \mathbf{q}$ is orthogonal to $\mathbf{q}$. There are many choices of the basis for these orthogonal vectors, and we choose the basis $\{\mathbf{B}_{1} \mathbf{q}, \mathbf{B}_{2} \mathbf{q},\mathbf{B}_{3} \mathbf{q}\}$ for our computations. Any arbitrary variation $\delta \mathbf{q} \in \mathbb{R}^{4}$ can be projected on $\{\mathbf{B}_{1} \mathbf{q}, \mathbf{B}_{2} \mathbf{q},\mathbf{B}_{3} \mathbf{q}\}$-space using the projections $\Lambda = \Pi^{T} \delta \mathbf{q} $ in $ \mathbb{R}^{3}$ where 
\begin{align*}
\Pi = \begin{bmatrix}
\mathbf{B}_{1} \mathbf{q} & \mathbf{B}_{2} \mathbf{q} & \mathbf{B}_{3} \mathbf{q} 
\end{bmatrix}
\in \mathbb{R}^{4 \times 3}.
\end{align*}
As a result, the new projection of the second variation reads as
\begin{align*}
\delta^2 J[\Lambda]= \frac{1}{2} \int_{0}^{l} \bigg[ \boldsymbol \Lambda^{\prime} \cdot \bar{\mathbf{P}} \boldsymbol \Lambda^{\prime} + \boldsymbol \Lambda \cdot \bar{\mathbf{C} } \boldsymbol \Lambda^{\prime} +
 \boldsymbol \Lambda^{\prime} \cdot \bar{\mathbf{C}}^{T}  \boldsymbol \Lambda  +  \boldsymbol \Lambda  \cdot \bar{\mathbf{Q}}^{T}  \boldsymbol \Lambda  \bigg] ds, 
\end{align*}
where 
\begin{align*}
    \bar{\mathbf{P}} = \Pi^{T} L_{\mathbf{q}^{\prime} \mathbf{q}^{\prime} }\Pi,\qquad
    \bar{\mathbf{Q}} = \left(\Pi^{\prime} \right)^{T}  L_{\mathbf{q}^{\prime} \mathbf{q}^{\prime} }\Pi^{\prime},\qquad
    \bar{\mathbf{C}} = \left(\Pi^{\prime} \right)^{T} L_{\mathbf{q}^{\prime} \mathbf{q}^{\prime} }\Pi + 
    \Pi^{T} L_{\mathbf{q} \mathbf{q}^{\prime} }\Pi.
\end{align*}
Along the projected directions, the Hessian matrix 
\begin{align}
     \bar{\mathbf{P}}= 
    \begin{bmatrix}
    K_{1} & 0 & 0\\
    0 & K_{2}  & 0\\
   0 & 0 & K_{3}\\
    \end{bmatrix},
\end{align}
is positive definite and satisfies Legendre's condition, whereas $L_{\mathbf{q}^{\prime} \mathbf{q}^{\prime} }$ alone is only positive semi-definite. The linearization of the Hamiltonian form of the equilibria~\eqref{eqn:Ham_rod_equilibria} gives the Hamiltonian form of the Jacobi operator $\mathcal{S}$
\begin{align}
\label{eqn:Jacobi_equations}
    \begin{bmatrix}
    \delta \mathbf{q}\\
    \delta \boldsymbol \mu
    \end{bmatrix}^{\prime}=
        \begin{bmatrix}
    \mathbf{O} &  \mathbf{I} \\
     -\mathbf{I}&  \mathbf{O} 
    \end{bmatrix} 
        \begin{bmatrix}
    H_{q q} &  H_{q \mu} \\
     H_{\mu q}&  H_{\mu \mu}
    \end{bmatrix} 
    \begin{bmatrix}
    \delta \mathbf{q}\\
    \delta \boldsymbol \mu
    \end{bmatrix},
\end{align}
where the Hessian matrices $H_{q q}$, $H_{q \mu}$,
     $H_{\mu q}$, $ H_{\mu \mu}$ are partial derivatives of $H$ with respect to respective arguments. By restricting the variations $\delta \mathbf{q}$ only to $\{\mathbf{B}_{1}\mathbf{q}, \mathbf{B}_{2}\mathbf{q},\mathbf{B}_{3}\mathbf{q}\}$ basis, we obtain the Hamiltonian version of $\mathcal{S}$ in projected space on $\Pi$. The absence of conjugate points in the interval $[0,l]$ is the sufficient condition for the equilibria $\zeta_{o}$ to be stable, the computation  of which is outlined below. The boundary with the natural boundary conditions, i.e., $(s=l)$ is chosen and IVP is solved towards the other end $(s=0)$ for a basis of initial values for $\delta \mathbf{q}$
\begin{align}
\label{eqn:Jacobi_equations_BC_l}
    \delta \mathbf{q}(l)=\mathbf{B}_{i}\mathbf{q}(l),  \quad i=1,2,3,
\end{align}
and with initial values of $\delta \boldsymbol \mu(l)$ that satisfy the linearized boundary condition~\eqref{eqn:Rod_bcs_Nat_b},~\eqref{eqn:Rod_bcs_Nat_c} 
\begin{subequations}
\begin{align}
    \delta \boldsymbol \mu(l) &\cdot \mathbf{B}_{i}\mathbf{q}(l )  + \boldsymbol \mu(l)  \cdot \mathbf{B}_{i} \delta \mathbf{q}(l) + \frac{\partial}{\partial \mathbf{q}}\left(\Upsilon(l) \times \mathbf{F} \right)_{i} \delta \mathbf{q}(l )=0, \quad i=1,2,3,\\
    \boldsymbol \mu(l) &\cdot  \delta \mathbf{q}(l) +   \mathbf{q}(l)  \cdot \delta \boldsymbol \mu(l)=0.
\end{align}
 \label{eqn:linearized_bcs}
\end{subequations}
The algebraic system \eqref{eqn:linearized_bcs} is solved to obtain the respective values of the components $\delta \boldsymbol \mu (l)$.\\

The components corresponding to the IVP solution $\delta \mathbf{q}(s)$ for the $i$ th set of IVP are denoted $\delta \mathbf{q}^{(i)}$, $i=1,2,3$. These four components are arranged as rows in the $4\times 3$ matrix along $s$ as
\begin{align}
\begin{bmatrix}
\delta \mathbf{q}^{(1)}(s) & \delta \mathbf{q}^{(2)}(s) & \delta \mathbf{q}^{(3)}(s) 
\end{bmatrix}.
\end{align}
This matrix is projected onto $\{\mathbf{B}_{1}\mathbf{q}, \mathbf{B}_{2}\mathbf{q},\mathbf{B}_{3}\mathbf{q}\}$- space to yield  a $ 3 \times 3 $ matrix. We call this matrix \emph{stability  matrix}. A point $l^{*}$ is called conjugate point of $l$ if the determinant of this $3 \times 3$ \emph{stability  matrix} vanishes for any $ l^{*} \in [0,l]$. Therefore, by Jacobi condition, if the equilibrium possesses a conjugate point computed through the above method, it is unstable. A system of 8-dimensional Jacobi equations~\eqref{eqn:Jacobi_equations} must be solved with three sets of boundary conditions~\eqref{eqn:Jacobi_equations_BC_l},~\eqref{eqn:linearized_bcs}, resulting in a 24-dimensional IVP to assess the stability of the equilibria determined by the $14$-dimensional boundary value problem (BVP)~\eqref{eqn:Ham_rod_equilibria}, \eqref{eqn:Rod_bcs_Dir} and \eqref{eqn:Rod_bcs_Nat}.

\begin{figure}[b!]
    \centering
    \includegraphics[width=0.85\textwidth]{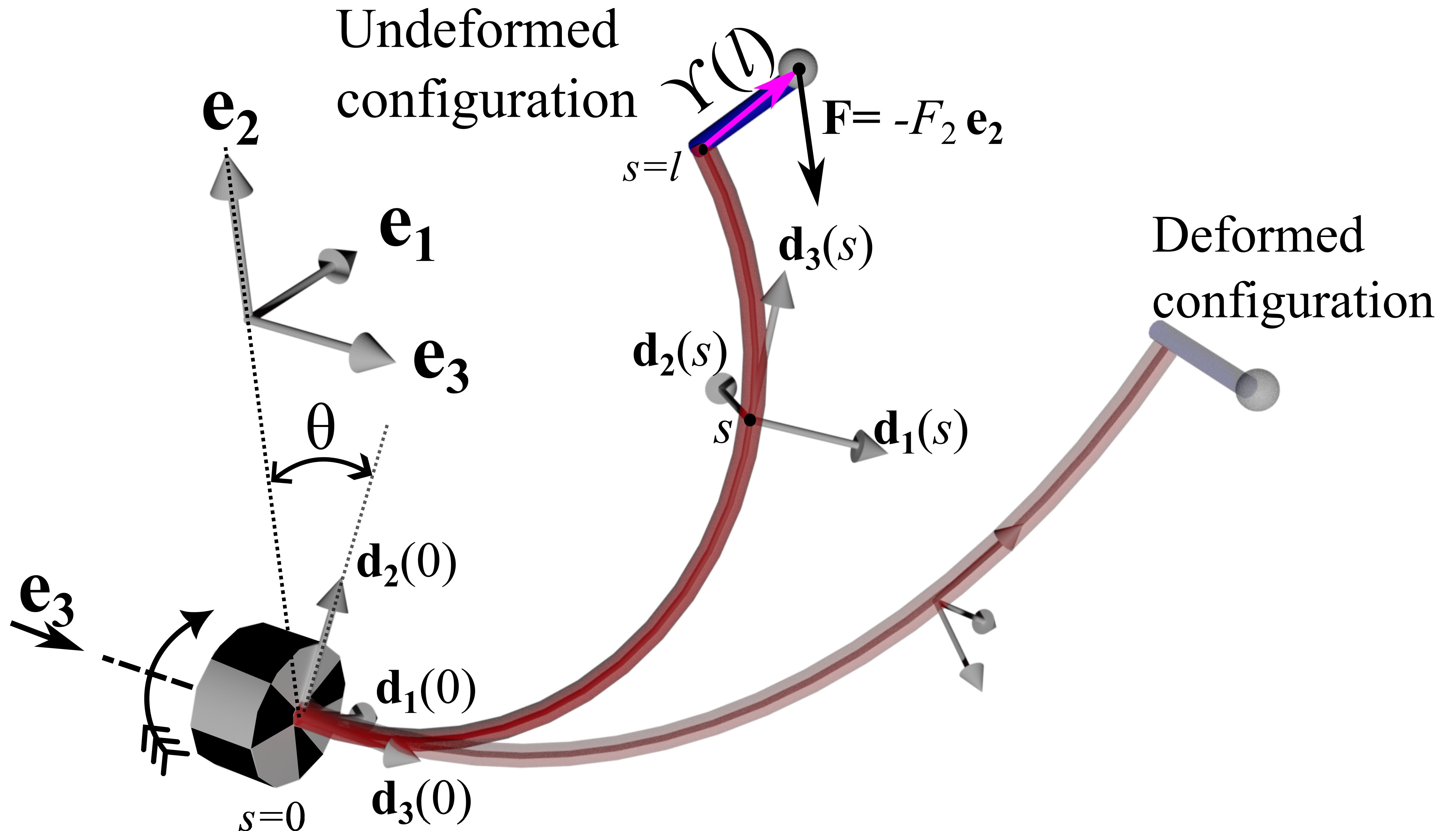}
    \caption{The schematic of the tip-loaded cantilever setup. One end of the rod is fixed to the quasi-statically rotating shaft, while the other end is attached to a load through a massless rigid lever arm (indicated in blue).}
    \label{fig:figure3}
\end{figure}

\section{Numerical Examples}
\label{sec:s4}
In this section, we apply the proposed conjugate point test to determine the stability of tip-loaded cantilever systems and investigate the potential occurrence of snap-back instability. Consider a naturally curved, slender, massless elastic rod clamped to a horizontal $\mathbf{e}_{3}-$axis at one end and attached to a dead load at the other end, with gravity acting in the vertical direction, as depicted in Figure~\ref{fig:figure3}. In all our examples, we assume uniform stiffness along the rod, with its components given by:
\begin{align}
K_{1}=K_{2}=EI=1.0, 
\qquad
K_{3}=\frac{EI}{1+\nu} , \quad \nu=0.5,
\end{align}
and a vertical tip load $\mathbf{F}=-F_{2} \mathbf{e}_{2}$. The clamped end is quasi-statically rotated about the horizontal tangent  $\mathbf{e}_{3}$-axis by setting the boundary condition
\begin{align}
\mathbf{q}(0)= \bigg[ 0 , 0, \sin \frac{\theta}{2}, \cos \frac{\theta}{2}\bigg]^{T},
\end{align}
and is numerically simulated by performing parameter continuation of~\eqref{eqn:Ham_rod_equilibria} in $\theta$ using AUTO-07p~\cite{doedel2007auto}. In this case, the system is $2 \pi-$ periodic with respect to $\theta$, i.e., for any $z \in \mathbb{R}$ the systems at $\theta=z$ and $\theta=  z + 2 \pi $ show identical features.

\subsection{Intrinsically Planar Curvatures}

\begin{figure}[b!]
\centering
    \centering
    \includegraphics[width=1.0\textwidth]{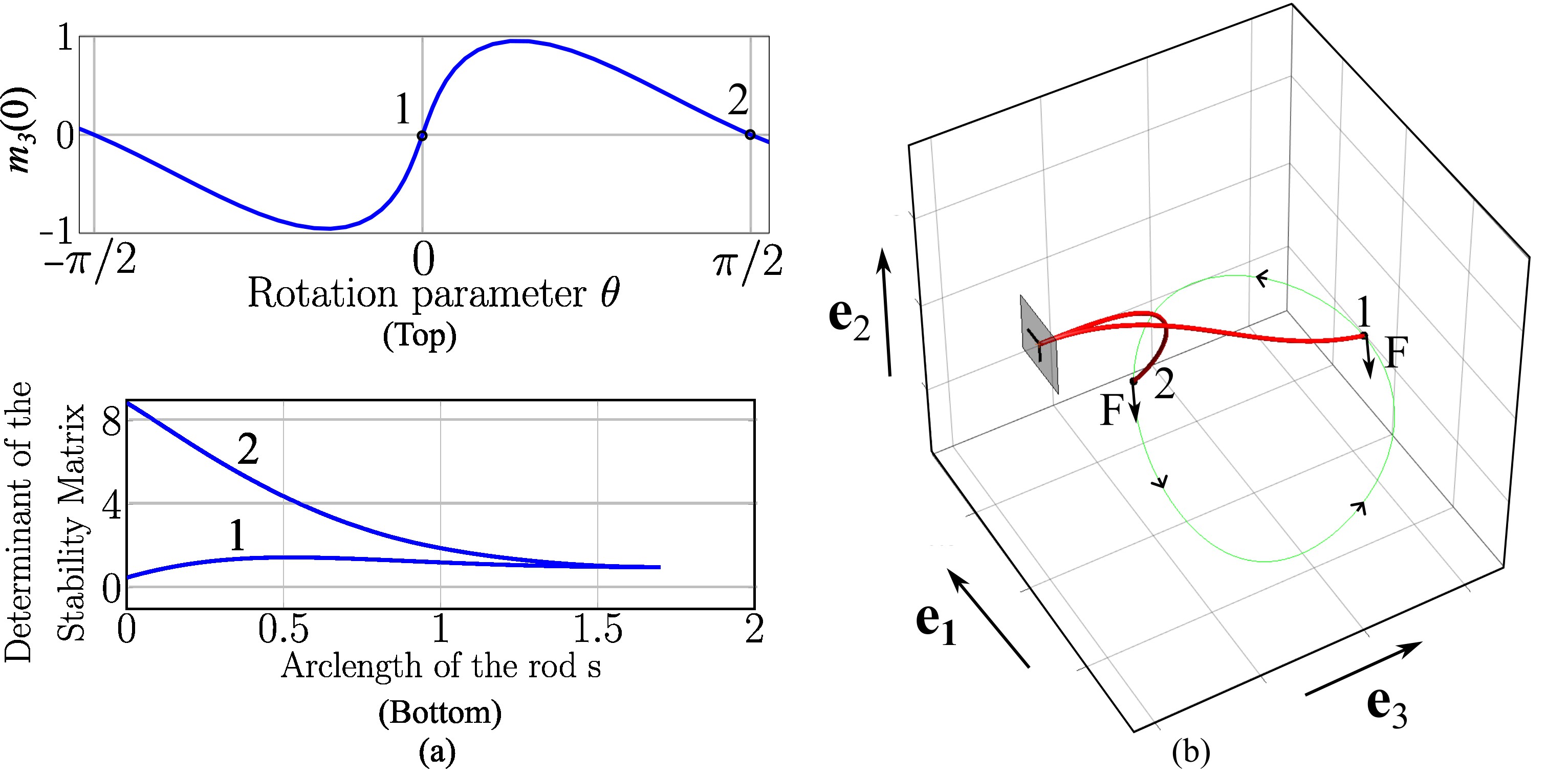}
 \caption{(a)(Top) The bifurcation diagram depicting the twist moment at the clamped end $m_{3}(0)$ as $\theta$ is varied. The equilibria corresponding to $\theta=0$ and $\theta=\pi/2 $ are chosen for stability analysis. (Bottom) The conjugate point computations for equilibria with labels $1$ and $2$. The determinant never vanishes, indicating the absence of conjugate points, and therefore, both equilibria are stable. (b) The tip trace of the cantilever as $\theta$ is varied and the centrelines for the equilibria with labels $1$ and $2$. The tip load is represented by a solid dot, with its direction indicated by arrows.}
 \label{fig:figure4}
\end{figure}

Initially, we analyze naturally curved planar rods that have intrinsic curvature of the form $\mathtt{\hat{u}}=(\hat{u}_{1}, 0, 0)$, subjected to a concentrated tip load $\Delta=(0,0,0)$. We quasi-statically rotate the clamped end by varying $\theta$ from $-\pi$ to $\pi$ and observe the resulting equilibria. This observation is repeated for increasing values of $\hat{u}_{1}$, length $l$, and tip load $F_{2}$. It is well known that if the rod is isotropic $(\hat{u}_{1}=0)$, the centerlines of the deformed configuration take the same shape for any rotation of the clamped end $\theta$, under a given non-zero tip load $F_{2}$. However, when symmetry is broken by introducing a non-zero $\hat{u}_{1}$, the rod centerlines assume different forms depending on $\theta$, even under the same applied load $F_{2}$. A few characteristics of the equilibria in this scenario are shown in Figure~\ref{fig:figure4} and~\ref{fig:figure5}, with the other parameters set at $l=1.7$ and $F_{2}=1.5$. For lower curvatures, such as $\hat{u}_{1}=1.0$, a single rod equilibrium exists for each $\theta$. The twist moment at the clamped end, $m_{3}(0)$, is evaluated from the continued solutions and plotted against $\theta$ to generate the bifurcation diagram. In our examples, $\theta=0$ corresponds to rod shapes curving vertically upward. The stability of these equilibria is assessed through conjugate point computations, as detailed in Section~\ref{sec:Elastic_rods_Jacobi_condition}. All the given equilibria are stable, as they exhibit no conjugate points. Some computations for the labelled equilibria are displayed in Figure~\ref{fig:figure4}a(Bottom). On the other hand, higher curvatures, such as $\hat{u}_{1}=1.5$, result in multiple equilibria for the same value of $\theta$ around $\theta=0$, as reflected by the folds in the parameter (in this case, $\theta$), as shown in Figure~\ref{fig:figure5}a(Top). In this region, three equilibria exist for a given $\theta$. Conjugate point tests reveal that two of these equilibria are stable, while the other (that lies between the folds) is unstable. A few rod configurations along with the tip trace are displayed in Figure~\ref{fig:figure5}b. The tip trace associated with the unstable equilibria is denoted by black dotted lines, indicating a discontinuous path of stable solutions. Consequently, the configurations abruptly transition from one stable configuration to another when operated around this parameter space (around $\theta=0$), mimicking a catapult-like behavior. The dynamics of a rotating cantilever, as a consequence of this instability, is beyond the scope of this paper. For further details on the snapping dynamics of elastic rods, refer to Armanini et al.~\cite{Armanini2017}. The standard bifurcation theory~\cite{Golubitsky1985,iooss2012elementary} predicts that the folds in the bifurcation parameter are the points of stability exchange, and our conjugate point computations concur with this. The evolution of the rod configurations and the snap-back instability depend on their history in cases with folds, as denoted by the arrows in Figure~\ref{fig:figure5}a (Top). Hence, we also use the term \emph{hysteresis} to describe this phenomenon.

\begin{figure}[t!]
\centering
    \centering
    \includegraphics[width=\textwidth]{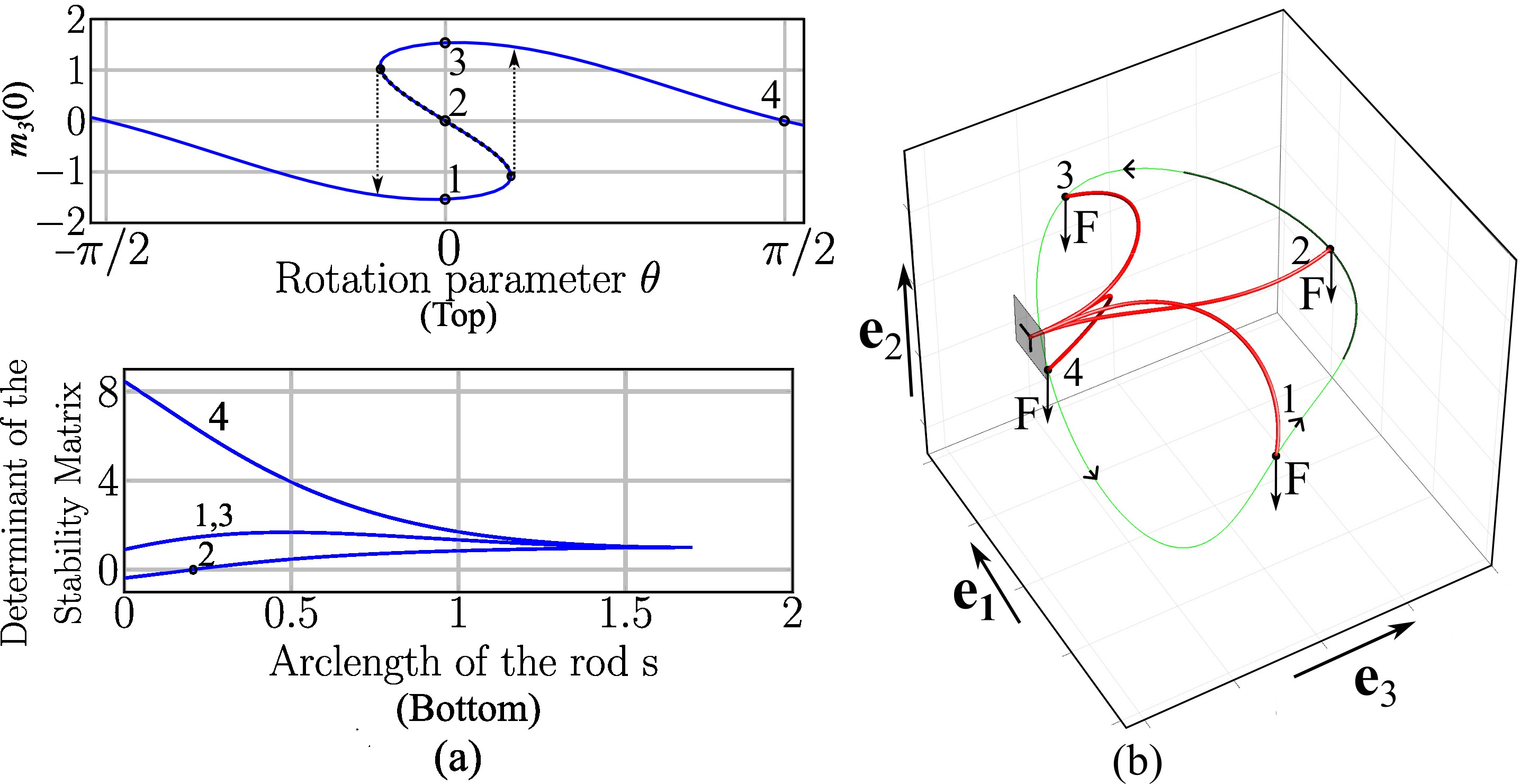}
     \label{fig:figure5b}
 \caption{(a) (Top) The bifurcation diagram depicting the twist moment at the clamped end $m_{3}(0)$ as $\theta$ is varied.  The equilibria corresponding to $\theta=0$ and $\theta=\pi/2 $ are selected for stability analysis. The plot has two folds and three equilibria exist for $\theta=0$. (Bottom) The conjugate point computations for the selected configurations. The determinant corresponding to the equilibrium $2$ vanishes indicating the presence of a conjugate point and is unstable. The remaining equilibria have no conjugate points and are stable. (b) The tip trace during the control maneuver and the rod centrelines of the selected equilibria. The tip load is represented by a solid dot, with its direction indicated by arrows. The tip corresponding to the equilibria lying between the folds is indicated by the black dotted line. The equilibria labelled $1$ and $3$ are mirror images about the $\mathbf{e}_{2} - \mathbf{e}_{3}$ plane and are nearly identical, which explains why the curves corresponding to conjugate tests in (a) bottom coincide.} 
 \label{fig:figure5}
\end{figure}

\begin{figure}[t!]
    \centering
    \includegraphics[width=1.0\textwidth]{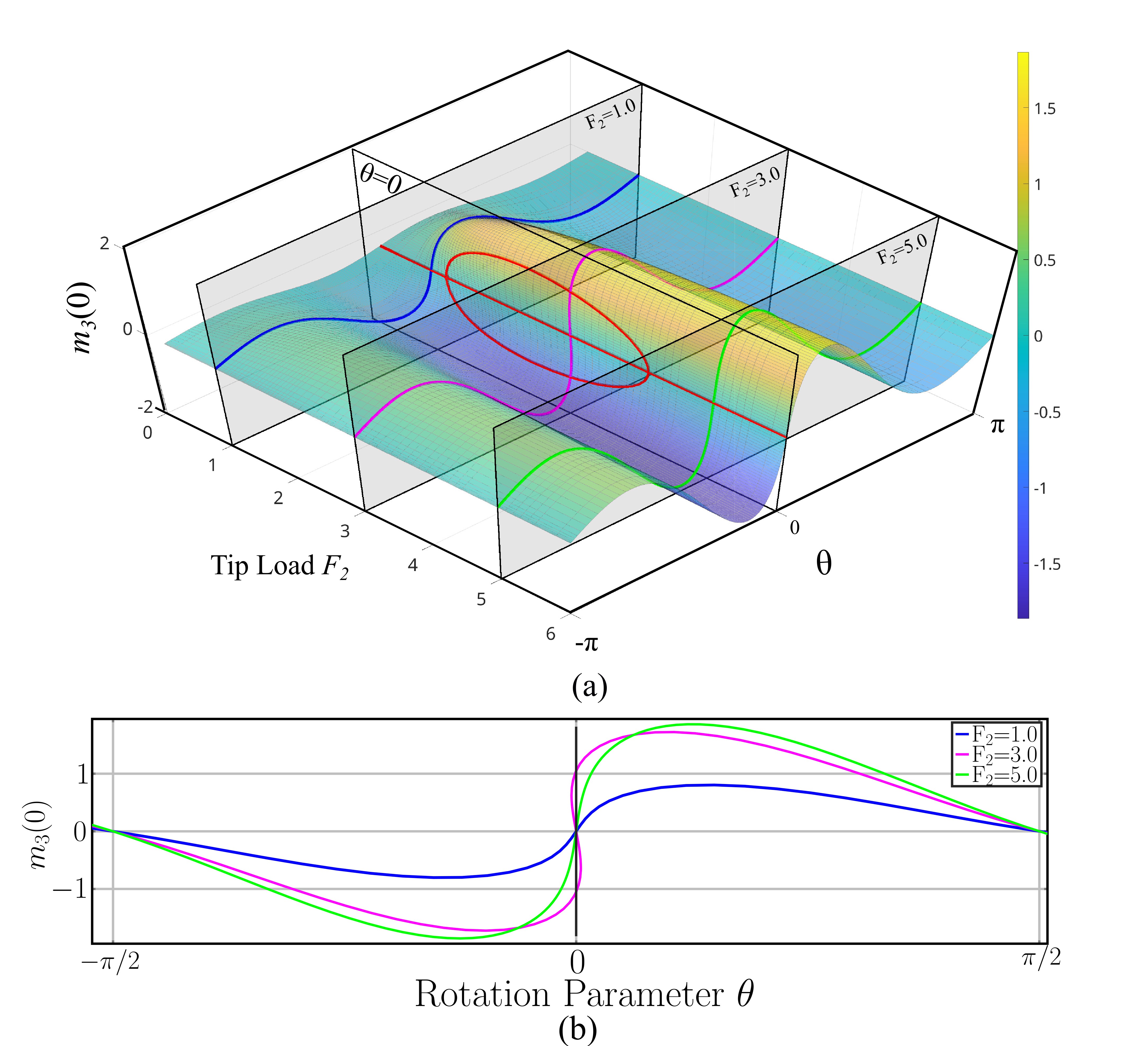}
        \caption{(a) The surface plot of $m_{3}(0)$ for continued solutions for an intrinsic curvature $\mathtt{\hat{u}}=(1.5,0,0)$ and length $l=1.4$ as a function of $\theta$ and tip load $F_{2}$. The planes $F_{2}=1.0,3.0$ and $5.0$ slice the surface giving smooth $m_{3}(0)$ vs. $\theta$ plots. (b) The planar bifurcation curves sliced along the planes.}
    \label{fig:figure6}    
\end{figure}

The nature and extent of the hysteresis region in cantilevers are governed by the complex interplay among system parameters such as intrinsic curvature $\hat{u}_{1}$, length $l$, and tip load $F_{2}$, as illustrated in the following example. Consider a cantilever setup with intrinsic curvature $\mathtt{\hat{u}} = (1.5,0,0)$, length $l = 1.4$, subjected to a concentrated tip load $(\boldsymbol \Delta = (0,0,0))$. We perform parameter continuation along the clamp angle $\theta$ from $-\pi$ to $\pi$ for a fixed $F_{2}$, with values ranging from $0$ to $6.0$. From the resulting solutions, bifurcation diagrams ($m_{3}(0)$ vs. $\theta$ plots) are generated for each $F_{2}$ and plotted as a surface plot, as shown in Figure~\ref{fig:figure6}a. This surface plot is known as the \emph{bifurcation surface}. The planar bifurcation diagrams for $F_{2}=1.0,3.0$ and $5.0$ are depicted by the corresponding $F_{2}-$ planes that bisect this surface, as shown in Figure~\ref{fig:figure6}b. The curves for $F_{2}=1.0$ and $F_{2}=5.0$ do not indicate any regions of unstable equilibria, while the curve for $F_{2}=3.0$ exhibits a region of unstable equilibria characterized by the presence of folds. This horizontal distance between the folds is smaller in this case compared to that in Figure~\ref{fig:figure5}a (Top). Another orthogonal plane $\theta=0$ bisects this surface, producing a curve that can be interpreted as a bifurcation diagram when the parameter $F_{2}$ is varied while $\theta$ is fixed at $0$. The presence of two pitchforks illustrates the symmetrical nature of the rod deformations around $\theta=0$. This diagram indicates the rod's response as the tip load $F_{2}$ increases, while $\theta$ is fixed at $0$, and it corresponds to the rod being planar and curving upward. We draw some preliminary conclusions, relying primarily on the plots and without extensive analysis. As the magnitude of the tip load $F_{2}$ increases, the planar equilibrium of the rod, represented by a straight line, undergoes two pitchfork bifurcations: the first is supercritical and the second is subcritical. According to bifurcation theory~\cite{Golubitsky1985, iooss2012elementary}, the trivial solution (which in this context is the planar rod equilibrium at $\theta=0$) loses stability at the supercritical pitchfork bifurcation as it passes through and regains its stability at the subcritical pitchfork bifurcation. The equilibria at $\theta=0$ that lie between the folds correspond to this unstable trivial line between these bifurcations. 

To gain better insight into the dependence of hysteresis behavior on system properties, we non-dimensionalize the quantities $s, \mathbf{m},\mathbf{n}, \mathbf{u} $ by substituting $s=\bar{s}/\hat{u}_{1}$, $\mathbf{\hat{u}}=\hat{u}_{1} \mathbf{\bar{u}}, K_{i}=EI \bar{K}_{i}$, and $ m_{i}= EI \hat{u}_{1} \tilde{m}_{i} $ for $i=1,2,3$, transforming~\eqref{eqn:Ham_rod_equilibria_c}, ~\eqref{eqn:Ham_rod_equilibria_d} to
\begin{subequations}
\begin{align}
 \frac{d \mathbf{q}}{d \bar{s}}&=  \sum_{j=1}^{3}\left( \bar{K}_{j}^{-1}m_{j} + \bar{u}_{j} \right)\frac{1}{2}\mathbf{B}_{j}\mathbf{q}, \\
		\frac{d \tilde{\boldsymbol \mu}}{d \bar{s}}  & =  \sum_{j=1}^{3} \left( \bar{K}_{j}^{-1}m_{j} + \bar{u}_{j} \right)\frac{1}{2} \mathbf{B}_{j}  \tilde{\boldsymbol \mu}   -\frac{\partial \mathbf{d}_{3}}{\partial \mathbf{q}} ^{T}  \bar{\boldsymbol \Gamma} / \phi^{2},
\end{align}
\end{subequations}
where $\bar{\boldsymbol \Gamma}= \mathbf{F}l^{2}/EI$ and $\phi=l \hat{u}_{1}$ are the associated dimensionless system parameters. In the present cantilever setup with intrinsic planar curvatures under a vertical concentrated tip load, the vector $\bar{\boldsymbol \Gamma}$ is given by $\left[ 0, -\Gamma, 0 \right]^{T} $, where $\Gamma= \frac{F_{2} l^{2}}{E I }$ and the dimensionless intrinsic curvature $\mathbf{\bar{u}}$ is given by $(1,0,0)$. The boundary conditions are
\begin{subequations}
\begin{align}
\mathbf{q}(0)= \bigg[ 0 , 0, \sin \frac{\theta}{2}, \cos \frac{\theta}{2}\bigg]^{T}, \qquad 
 \tilde{m}_{i}(l)= \tilde{\boldsymbol \mu}(l) \cdot \mathbf{B}_{i} \mathbf{q}(l)/2 =0, \quad i=1,2,3.
\end{align}
\end{subequations}

\begin{figure}[t]
    \centering
    \includegraphics[width=0.99\textwidth]{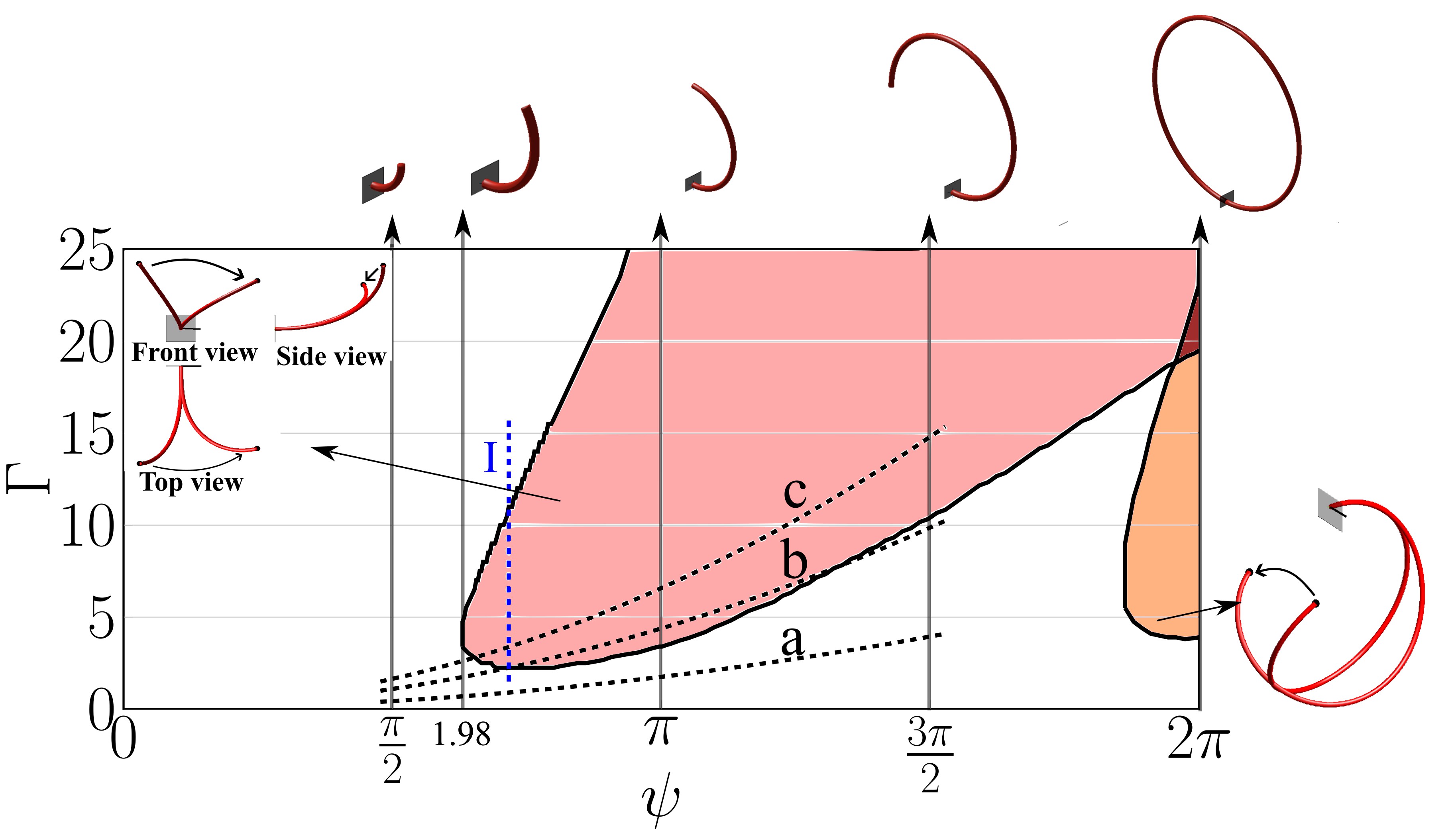}
    \caption{The $\Gamma - \psi $ region over which the snap-back instability arises when the clamped end is rotated from $-\pi$ to $\pi$. The undeformed shapes corresponding to different $\psi$ are indicated above the plot. The paths taken by different control maneuver namely tuning $F_{2}$ at fixed $l$ is denoted by curve I and tuning $l$ at fixed $F_{2}$ denoted by the curves a,b and c. In most of the shaded region, snap-back instability occurs around $\theta=0$. In a smaller portion near $\psi=2\pi$ (depicted in a different shade), snap-back instability occurs around $\theta=\pi$. Illustrations of the snapping motion for these cases are also indicated.}
    \label{fig:figure7}
\end{figure}
 \begin{figure}[t]
    \centering
     \includegraphics[width=\textwidth]{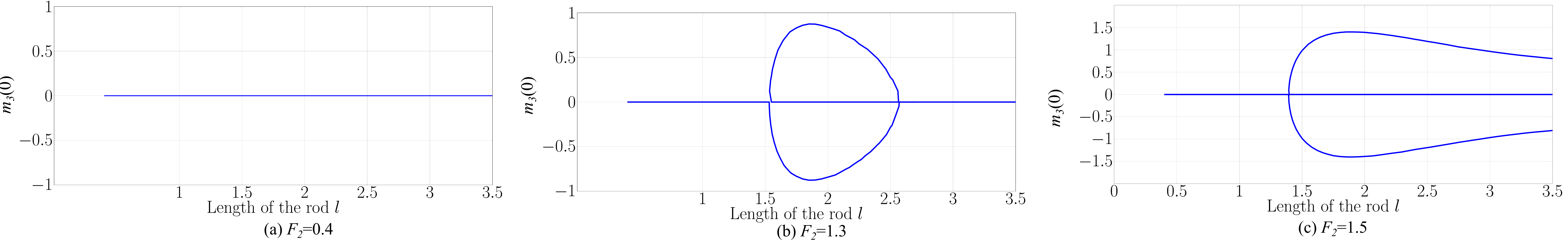}
    \caption{Bifurcation diagrams on the plane $\theta=0$ for increasing values of $l$ at different tip loads.}
    \label{fig:figure8}
\end{figure}

 The equation~\eqref{eqn:Ham_rod_equilibria_a} relates the rod's centerline $\mathbf{r}$ with quaternions $\mathbf{q}$, while the other term~\eqref{eqn:Ham_rod_equilibria_b} yields the constant internal force $(\mathbf{n}(s)=[0,-F_{2},0]^{T})$, and their governing ODEs are disregarded for this analysis. The parameter $\Gamma$ represents the scale of the cantilever system and the applied load, whereas $\psi$ represents the dimensionless curvature and is the angle formed by the circular arc at its center. We perform a continuation in the parameter $\theta$ for a complete rotation, considering different values of $\psi$ and $\Gamma$ and assess if the hysteresis region exists. Figure~\ref{fig:figure7} depicts the $\Gamma - \psi$ space where the hysteresis region for a rotating cantilever exists. In this analysis, we restrict our consideration to values of $\psi$ up to $2\pi$, representing a complete circle turn, while neglecting any instances of self-contact. The unstable modes emerge only for values of $\psi > 1.985 $. The shape of this diagram is influenced by the Poisson ratio $\nu$ of the material, which is set to an incompressible case of $0.5$ in this study. There are two distinct hysteresis regions, each represented by different shades. One region, occupying the majority of the plot, corresponds to the unstable equilibria near $\theta=0$, as indicated. The other region, a small portion located around $\psi=2\pi$, represents the unstable equilibria around $\theta=\pi$. There is a small portion on the top right where unstable equilibria occur around both $\theta=0$ and $\theta=\pi$. A key takeaway from this diagram is that there are areas where the hysteresis can be manipulated by adjusting some parameters in both upward and downward directions. One example, discussed earlier in Figure~\ref{fig:figure6}, shows how hysteresis can be adjusted by increasing or decreasing the value of $F_{2}$, with the path taken by this case at fixed $l$ indicated by the vertical line I on the $\Gamma-\psi$ plot. Another control instance involves keeping $F_{2}$ fixed while varying the length $l$. To explore this approach, we analyze the bifurcation diagrams obtained by slicing the bifurcation surface at $\theta = 0$ plane for different values of $F_{2}$, as shown in Figure~\ref{fig:figure8}. For smaller tip loads, such as $F_{2}=0.4$, no hysteresis is observed for $l \in [0.5,3.5]$. As $F_{2}$ is increased to $1.3$, hysteresis behavior emerges intermediate values of $l$. A further increase in $F_{2}$ to $1.5$ results in hysteresis for all $l > 1.4$. These control paths take the form of parabolas $(\Gamma= \frac{F_{2}l^{2}}{EI})$, as indicated by the dotted curves a,b, and c in the $\Gamma-\psi$ plot in Figure~\ref{fig:figure7}. In conclusion, the tip load $F_{2}$ can be adjusted to create various control scenarios, where the hysteresis region either does not appear, appears only for intermediate values, or emerges for all values exceeding a critical threshold. This selective range of values, for which hysteresis can be switched on or off, holds significant potential for applications in the design of soft robotic arms. From an engineering perspective, the parameters $l$ or $F_{2}$ can be externally tuned. Another example of a control mechanism is by tuning the intrinsic curvature, which is feasible in active elastic rods~\cite{KACZMARSKI2022104918, Kazmarcki2024}.

\begin{figure}[t!]
 \centering
    \includegraphics[width=0.85\textwidth]{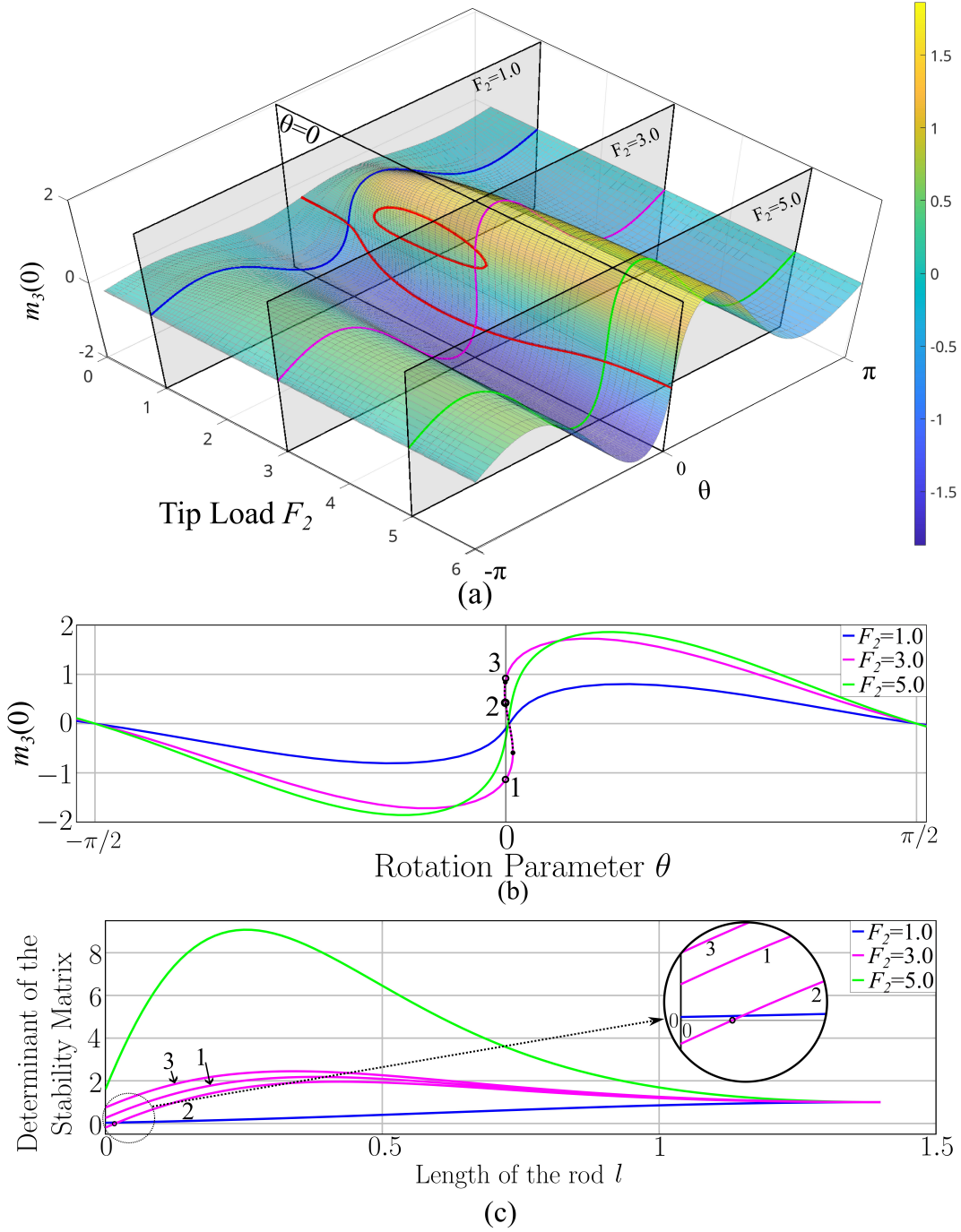}
 \caption{The torsion $\hat{u}_{3}$ acts as a symmetry-breaking parameter. The bifurcation surface sliced by planes $\theta=0$, $F_{2}=1.0, F_{2}=3.0$, and $F_{2}=5.0$. (b) The bifurcation diagrams for $F_{2}=1.0$ and $F_{2}=5.0$ have no folds, while that for $F_{2}=5.0$ has a fold. The equilibria corresponding to $\theta=0$ in all these cases are chosen for stability analysis, with labels assigned for the folded case of $F_{2}=3.0$. (c) Conjugate point computations for the equilibria at $\theta=0$. The determinant for the equilibrium with label $2$ vanishes, indicating a conjugate point, and therefore, is unstable. The remaining equilibria exhibit no conjugate points and are stable.}
 \label{fig:figure9}
\end{figure}

\subsection{Effect of Torsion Component or Load Arm}
Next, we examine the influence of other parameters, such as the torsion component $\hat{u}_{3}$ and the load arm $\boldsymbol \Delta$, on the hysteresis behavior. Initially, we tune the component $\hat{u}_{3}$ from $0$ to a small non-zero value of $0.01$, while keeping the remaining parameters fixed as in the prior case $(l=1.4, \hat{u}_{1}=1.5, \Delta=(0,0,0))$. An analysis is then performed by varying $\theta$ at incrementally increasing values of $F_{2} \in [0,6.0]$, and the associated features are displayed in Figure~\ref{fig:figure9}. The symmetric surface in Figure~\ref{fig:figure6} transforms into a non-symmetric bifurcation surface in Figure~\ref{fig:figure9}a, clearly evident when the $\theta=0$ plane  bisects it, revealing two disconnected, non-symmetric curves (in red). The planar bifurcation diagrams obtained by bisecting the planes $F_{2}=1.0,3.0$ and $5.0$ is shown in Figure~\ref{fig:figure9}b. The plots corresponding to $F_{2}=1.0$ and $F_{2}=5.0$ do not exhibit any folds, while the bifurcation diagram for $F_{2}=3.0$ features a fold. The equilibrium at $\theta=0$, placed between the folds (labeled $2$), is unstable as it possesses a conjugate point, as shown in Figure~\ref{fig:figure9}c. Unlike the previous case where $\hat{u}_{3}=0$ in Figure~\ref{fig:figure5}a, the conjugate point test plots associated with $\theta=0$, labelled $1$ and $3$ do not coincide here. This is because these configurations are no longer mirror images about $\theta=0$. Additionally, the position of the folds is slightly shifted to the right, and the $\theta=0$ line is no longer centered between the folds.

\begin{figure}[t!]
 \centering
 \includegraphics[width=\textwidth]{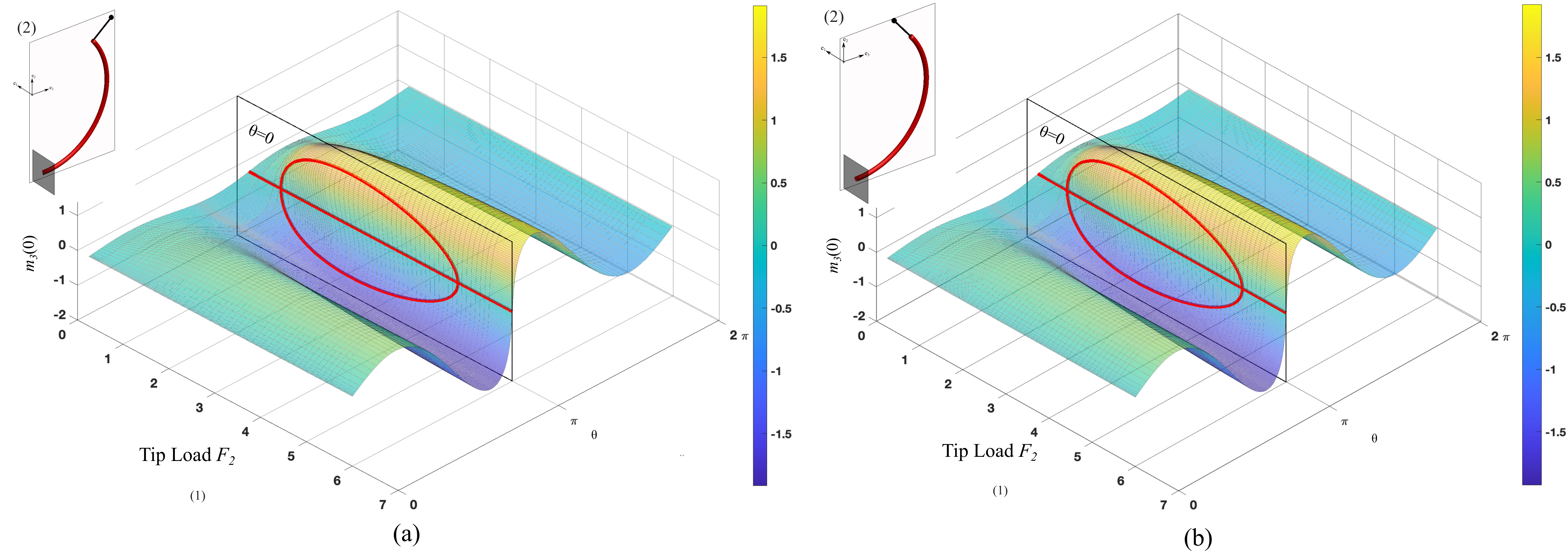}
 \caption{Bifurcation surfaces for non-zero in-plane arm components $\Delta_{2},\Delta_{3}$. (a) $l=1.4$, $\boldsymbol \Delta=(0,0.01,0)$ and $\mathtt{\hat{u}}=(1.5,0,0)$. (b)$ l=1.4$  $\boldsymbol \Delta=(0,0,0.01)$ and $\mathtt{\hat{u}}=(1.5,0,0)$. The schematic of elastic rods with the lever arm in the undeformed configuration is shown at the top left. The surface exhibits similar qualitative behavior to that of $\boldsymbol \Delta=(0,0,0)$.}
 \label{fig:figure10}
\end{figure}

\begin{figure}[t!]
 \centering
    \includegraphics[width=0.85\textwidth]{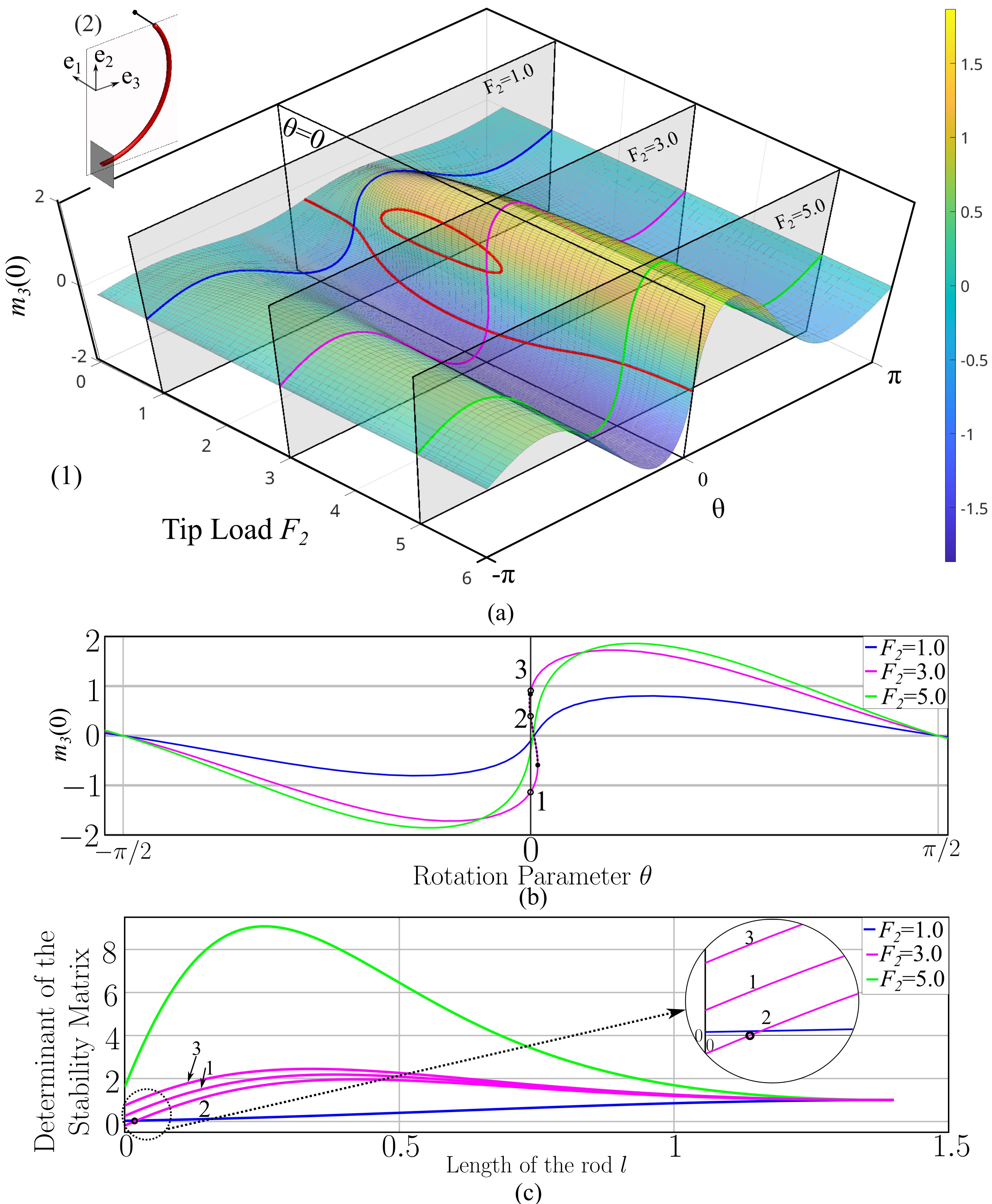}
 \caption{The non-planar arm component $\Delta_{1}$ acts as a symmetry-breaking agent. (a) The bifurcation surface sliced by planes at $\theta=0$, $F_{2}=1.0, F_{2}=3.0$, and $F_{2}=5.0$.  (b)  Planar bifurcation plots at $F_{2}=1.0,3.0$ and $5.0$. The plot for $F_{2}=3.0$ has two folds, while others are unfolded. (c) Conjugate point computations for the equilibria at $\theta=0$. The determinant for the equilibrium $2$ vanishes, indicating a conjugate point, and therefore, is unstable. Other equilibria have no conjugate points and are stable.}
 \label{fig:figure11}
\end{figure}

We now focus on the effect of the load arm $\boldsymbol \Delta$, which consists of three components $\Delta_{1},\Delta_{2}$ and $\Delta_{3}$. The centerlines of an unstressed uniform rod with intrinsic curvature $\mathtt{\hat{u}}=(\hat{u}_{1},0,0)$ are planar and lie in the $\mathbf{d}_{2}-\mathbf{d}_{3}$ plane. Consequently, a non-zero $\Delta_{2}$ or $\Delta_{3}$ results in the arm lying in this plane. These two parameters generate bifurcation surfaces, shown in Figure~\ref{fig:figure10}, which qualitatively resemble those in Figure~\ref{fig:figure6} and indicate symmetric behavior about $\theta=0$. In contrast, a non-zero $\Delta_{1}$ results in the arm non-planar with the undeformed rod, introducing asymmetry into the system, as demonstrated by its bifurcation surfaces corresponding to $\Delta_{1}=0.01$ in Figure~\ref{fig:figure11}. Figure~\ref{fig:figure12} presents a comparison among the cases $\boldsymbol \Delta=(\pm0.01,0,0)$, $\boldsymbol \Delta=(0,\pm0.01,0)$, and $\boldsymbol \Delta=(0,0,\pm0.01)$ using plots from the $\theta=0$ plane bisecting the bifurcation surfaces. In the present case, a positive component $\Delta_{2}$ shrinks the band of hysteresis (the distance between bifurcations), while a negative component expands it. On the other hand, a positive $\Delta_{3}$ shrinks the band, while its negative component expands it. From a technical perspective, a tunable arm $\boldsymbol \Delta$ can be easily designed and controlled using thermal expansion or linear motors, enabling the snapping behavior to be activated or deactivated as needed. The non-planar component $\Delta_{1}$ inverts the bifurcation diagram for its negative component. Thus, the load arm $\boldsymbol \Delta$ can stabilize or destabilize some cantilever equilibria. In conclusion, the parameters $\hat{u}_{3}$ and $\Delta_{1}$ induce symmetry-breaking in the bifurcation surfaces, while the components $\Delta_{2}$ and $\Delta_{3}$ quantitatively vary the bifurcation surfaces without altering their qualitative characteristics.

\begin{figure}[t!]
 \centering
    \includegraphics[width=\textwidth]{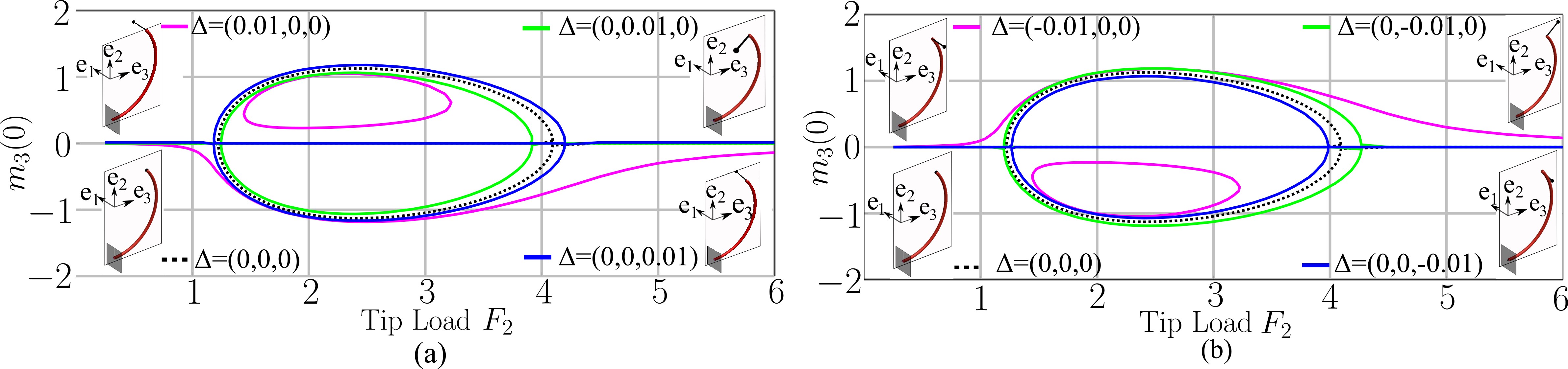}
 \caption{The effect of different load arm components $\Delta_{1},\Delta_{2}$ and $\Delta_{3}$ on the hysteresis behavior. The plots for the case of $\Delta=(0,0,0)$ are also shown in dotted lines. (a) Positive  $\Delta_{1},\Delta_{2}$ and $\Delta_{3}$. (b)Negative $\Delta_{1},\Delta_{2}$ and $\Delta_{3}$.}
 \label{fig:figure12}
\end{figure}

\begin{figure}[t!]
 \centering
    \includegraphics[width=\textwidth]{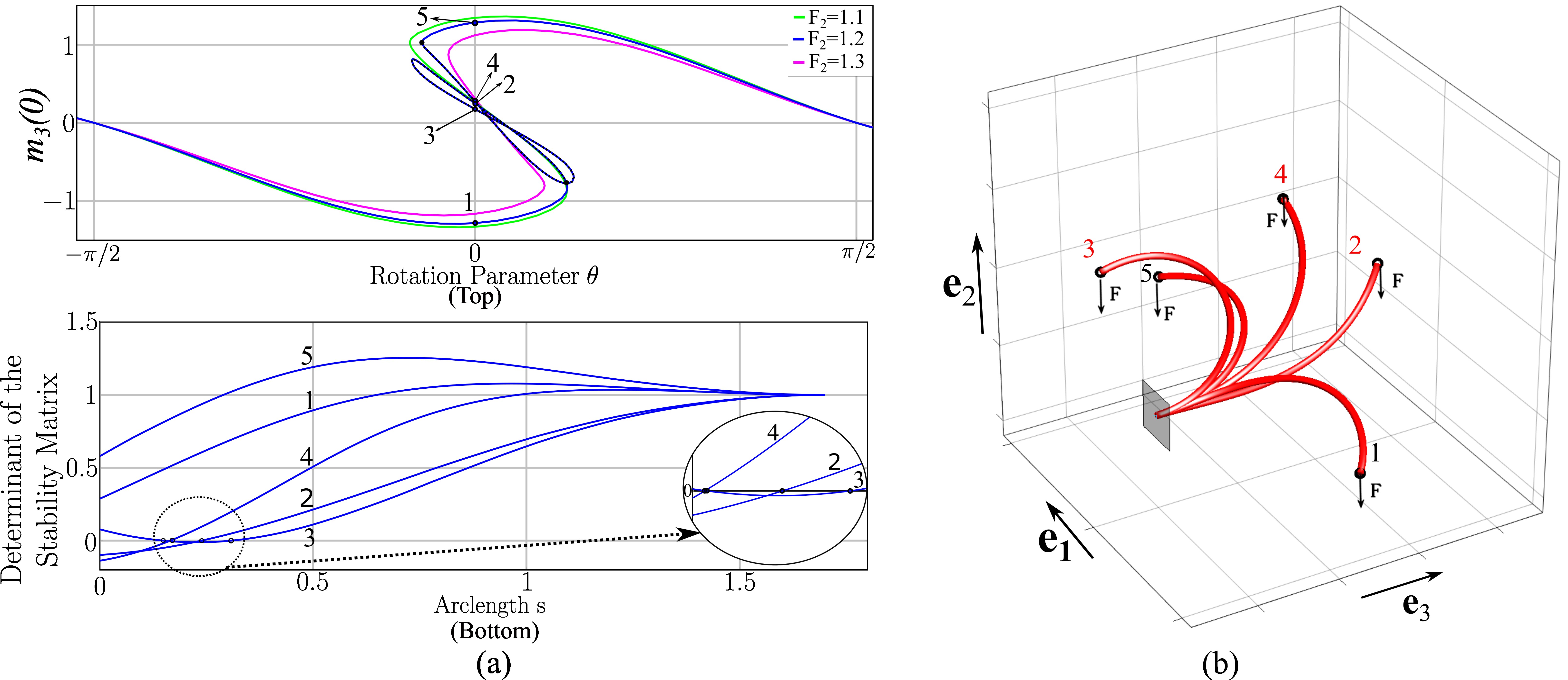}
    \caption{(a) (Top) Bifurcation diagrams for $\mathtt{\hat{u}}=(1.5,0,0.1)$,  $l=1.7$, and $\boldsymbol \Delta=(0,0,0)$ at $F_{2}=1.1,1.2$ and $1.3$. (Bottom) Conjugate point computations for equilibria at $\theta=0$ for $F_{2}=1.2$, which has a bifurcation diagram with four folds. The determinants corresponding to labels $2$ and $4$ vanish at one point (one conjugate point), while the determinant corresponding to $3$ vanishes twice (two conjugate points). The remaining equilibria ($1$ and $5$) have no conjugate points. (b) Equilibrium configurations corresponding to $\theta=0$ for $F_{2}=1.2$. The tip load is represented by a solid dot, with its direction indicated by arrows. The number of conjugate points is indicated adjacent to them. Only the equilibria with zero conjugate points (labels $1$ and $5$) are stable and can exist realistically.}
    \label{fig:figure13}
 \end{figure}

 \begin{figure}[t!]
 \centering
    \includegraphics[width=\textwidth]{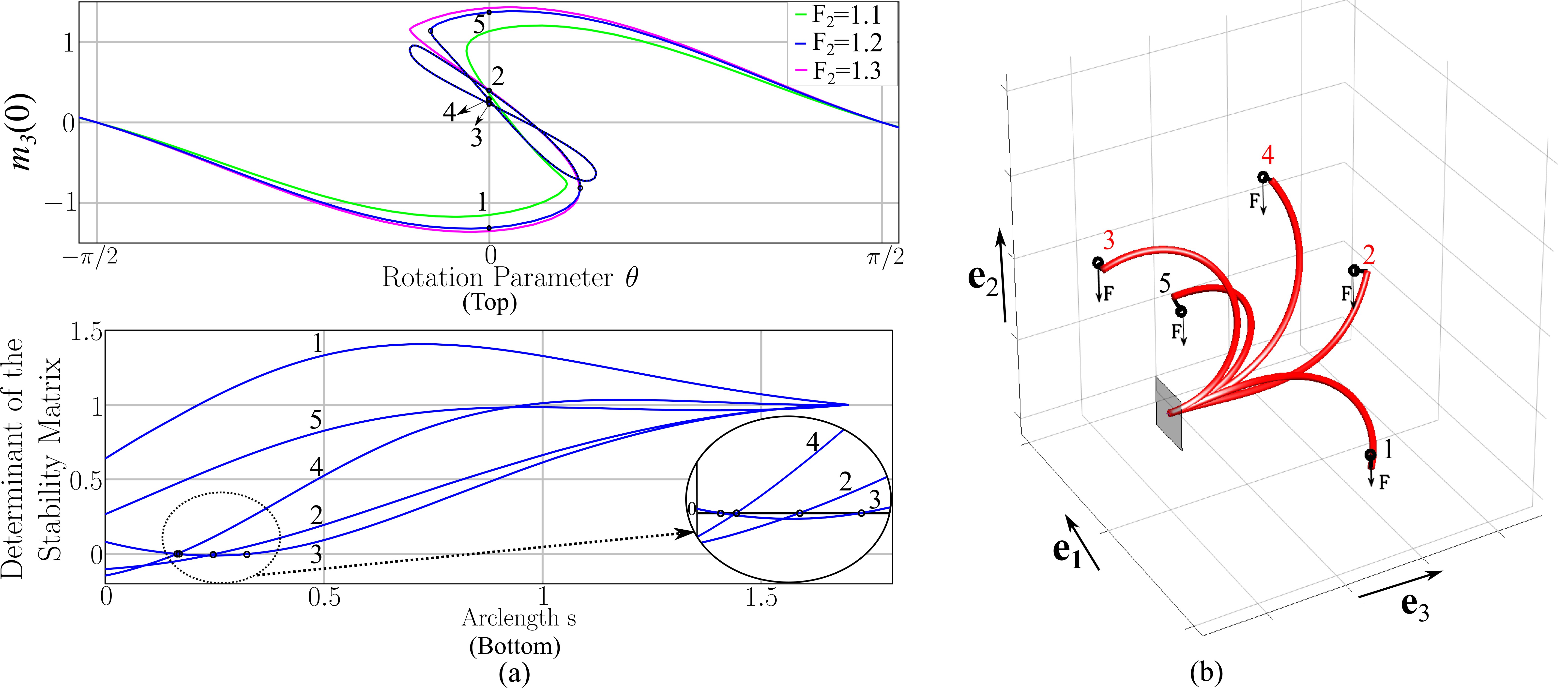}
    \caption{(a) (Top) Bifurcation diagrams for $\boldsymbol \Delta=(0.1,0,0)$, $l=1.7$, and $\mathtt{\hat{u}}=(1.5,0,0.0)$ at $F_{2}=1.1,1.2$ and $1.3$. (Bottom) Conjugate point computations for equilibria at $\theta=0$ for $F_{2}=1.2$ which has a bifurcation diagram with four folds. The determinants corresponding to labels $2$ and $4$ vanish at one point (one conjugate point), while the determinant corresponding to $3$ vanishes twice (two conjugate points). The remaining equilibria ($1$ and $5$) have no conjugate points. (b) Equilibrium configurations corresponding to the $\theta=0$ for $F_{2}=1.2$. The tip load is represented by a solid dot, with its direction indicated by arrows. The number of conjugate points is indicated adjacent to them. Only the equilibria with zero conjugate points (labels $1$ and $5$) are stable and can exist realistically.}
 \label{fig:figure14}
 \end{figure}

So far, our analysis has mainly focused on cases that exhibit either zero or one conjugate point. However, higher values of $l$, $\hat{u}_{3}$ and $\Delta_{1}$ may yield equilibria with more than one conjugate point. For example, let us increase $l$ to $1.7$ and $\hat{u}_{3}$ to $0.1$ and perform the continuation in $\theta$ for incrementing values of $F_{2}$ within the range $[0,6.0]$ at $\hat{u}_{1}=1.5$ and $\boldsymbol \Delta=(0,0,0)$. For intermediate values of $F_{2}$, we observe four folds, as depicted in Figure~\ref{fig:figure13}a. In this analysis, we focus solely on this region of interest and do not present the full bifurcation surface. The number of folds increases from two to four as $F_{2}$ increases from $1.1 $ to $1.2 $ and then reduces back to two when $F_{2}$ increases further to $1.3$. The five equilibrium configurations corresponding to $\theta=0$ on a plot with four folds are depicted in Figure~\ref{fig:figure13}b, along with their stability characteristics. The equilibrium labeled $3$ has two conjugate points, whereas the equilibria $2$ and $4$ have each one conjugate point and are unstable. The number of conjugate points increases or decreases by one at each fold. 

A similar analysis is conducted for parameters $l=1.7$ and $\Delta_{1}=0.1$, with continuation performed for incrementing values of $F_{2} \in [0,6.0]$ at $\mathtt{\hat{u}}=(1.5,0,0)$. The observed qualitative behavior remains consistent with the previous case. Five equilibrium configurations for $\theta = 0$, shown on the curve with four folds are analyzed for their stability, as depicted in Figure~\ref{fig:figure14}.  The parameters $\hat{u}_{3}$ and $\Delta_{1}$ exhibit similar qualitative effects on the hysteresis behavior of the cantilever system. Moreover, the last two examples in this section illustrate how certain parameters significantly influence the nonlinear behavior of cantilevers with intrinsic curvatures. Typically, near a fold, a stable equilibrium becomes unstable. However, an unstable equilibrium may either become stable or transition to a higher unstable mode with more than one conjugate point. In these two examples, we observe all possible stability transitions at the folds: from stable to unstable, from unstable to a higher unstable mode, back to a lower unstable mode, and then to stable at successive folds. There are studies relating the direction of these transitions near the folds through distinguished bifurcation diagrams~\cite{Maddocks1987, Hoffman2005}. However, these results are limited to cases with homogeneous boundary conditions or fixed-fixed boundary conditions, and their application to the current case of fixed-free boundary conditions requires further investigation. In our studies, conjugate point tests successfully captured stability transitions at folds.

\section{Summary and Discussion}
\label{sec:s5}
The Jacobi condition has been generalized to analyze the critical points of variational problems with fixed-free ends. The literature on the necessary and sufficient conditions for this set of problems is relatively sparse. For this analysis, the standard definition of conjugate points is slightly modified. This theory was developed keeping in mind the applications relevant to the rapidly advancing soft robots. The equilibria of tip-loaded cantilevers, which mimic flexible soft robotic arms, were computed using the Hamiltonian formulation, and their stability was analyzed by computing conjugate points. The Jacobi equations were shooted as IVPs from the free end towards the fixed end to evaluate conjugate points. The role of intrinsic curvature in generating the nonlinear behavior of elastic rods was particularly emphasized through numerical examples. A flexible, intrinsically curved elastic rod is subjected to a quasi-static rotation at one end and a vertical tip load at the other. Depending on system parameters, there are two possible outcomes: the tip either traces a smooth, continuous curve, or it traces a discontinuous curve due to intermediate unstable equilibria, exhibiting snap-back instability. Surprisingly, the hysteresis behavior displayed a complex dependence on tip load, length, and intrinsic curvature. For example, the hysteresis behavior displayed non-monotonic characteristics for a specific combination of parameters. An initial increase in tip load beyond a critical value led to the onset of hysteresis. But when the load is increased beyond a second critical value, the hysteresis behavior vanished. This intricate dependence on parameters was numerically represented through a non-dimensionless plot. Furthermore, the impact of the load arm in stabilizing or destabilizing
the rod equilibria was discussed. These findings are valuable for the design of innovative devices that can be used as switches or triggers when operated near snapping region. By employing functional materials in the cantilever structure, stimuli such as heat, light, chemicals, electric field, and magnetic field can be used to tune its parameters~\cite{shen2020stimuli}. This investigation can be further extended to cases involving distributed loads from gravity~\cite{Miller2014}, as well as multiphysics coupling effects due to light~\cite{Goriely2023}, electrostatics~\cite{singh2022computational}, and magnetism~\cite{avatar2024kirchhoff}. Additionally, isoperimetric constraints~\cite{Manning1998}, where the free end is fixed in position but free to rotate, may also be considered in future studies.

Generally, folds in continuation solutions indicate an exchange of stability. In all examples, the stability transitions at the folds align with the conjugate point tests. Although, the direction of change is unknown, the $2 \pi$ - periodicity of the cantilever system in the rotation parameter, and the information of folds may aid in predicting the stability when there are just two folds. However, one solution along the continued solutions must be analyzed for stability and should correspond to the stable equilibrium to effectively implement this strategy. Moreover, the stability prediction based solely on the folds may fail when more than two folds occur consecutively, as seen in some examples. In these scenarios, stability can be deduced through conjugate point tests. Distinguished bifurcation diagrams~\cite{Maddocks1987} are another useful tool to determine the direction of the stability change in variational problems, and they must be generalized to fixed-free boundary conditions for their use. Several studies have been conducted relating the number of conjugate points to the Morse index~\cite{Manning1998, Hoffman2002}, the maximal dimension of the subspace over which the second variation is negative definite. The combination of Legendre's strengthened condition, Sturm-Liouville problem, and Rayleigh quotients may allow this extension to the cases with fixed-free ends.

\section*{Acknowledgments}
I thank Prof. John Maddocks for fruitful discussions and for sharing his extensive knowledge of variational principles and elastic rods. This work would not have been possible without his guidance. I also thank Prof. Raushan Singh for his insightful comments on the initial drafts of this work. This study was funded by the Einstein Foundation Berlin and the Deutsche Forschungsgemeinschaft (DFG, German Research Foundation) under Germany's Excellence Strategy – The Berlin Mathematics Research Center MATH+ (EXC-2046/1, project ID: 390685689).

\enlargethispage{20pt}

\vskip2pc





\clearpage

\bibliographystyle{elsarticle-num}
\bibliography{Arxiv}







\end{document}